\documentclass[10pt,twocolumn,twoside]{IEEEtran}
\usepackage{amsbsy,amsmath,amsfonts,amssymb,amsbsy,subfigure,verbatim}
\usepackage{bm,cite,graphicx,psfrag,pstricks,theorem,tikz,times,url,verbatim}
\usepackage[english]{babel}
\usetikzlibrary{arrows,automata}
\newcommand{\bi}{\begin{itemize}}
\newcommand{\ei}{\end{itemize}}
\newcommand{\ben}{\begin{enumerate}}
\newcommand{\een}{\end{enumerate}}

\newcommand{\bc}{\begin{cases}}
\newcommand{\ec}{\end{cases}}
\newcommand{\bd}{\begin{description}}
\newcommand{\ed}{\end{description}}

\newcommand{\be}{\begin{equation}}
\newcommand{\ee}{\end{equation}}
\newcommand{\bea}{\begin{eqnarray}}
\newcommand{\eea}{\end{eqnarray}}

\newcommand{\bs}{\boldsymbol}

\newcommand\T{\rule{0pt}{2.0ex}}
\newcommand\B{\rule[-0.8ex]{0pt}{0pt}}

\newtheorem{thm}{Theorem}

\newtheorem{definition}{Definition}

\newtheorem{propos}{Proposition}

\newtheorem{ass}{Assumption}

\newtheorem{algo}{Algorithm}

\theoremstyle{plain}
\newtheorem{remark}{Remark}

\newcommand{\mbb}{\mathbb}
\newcommand{\mc}{\mathcal}
\newcommand{\mr}{\mathrm}

\newcommand{\nn}{\nonumber}

\graphicspath{%
     {./Figures/}
}

\begin{document}

\title{Optimal Adaptive Random Multiaccess in\\
 Energy Harvesting Wireless Sensor Networks
}
\author{
Nicol\`{o}~Michelusi,~\IEEEmembership{Member,~IEEE,}
        and~Michele~Zorzi,~\IEEEmembership{Fellow,~IEEE}
\thanks{The research of Nicol\`{o}~Michelusi has been funded in part by the following grants:
AFOSR FA9550-12-1-0215, ONR N00014-09-1-0700.
}
\thanks{N. Michelusi is with the Dept. of Electrical Engineering, University of Southern California. email: michelus@usc.edu.}
\thanks{M. Zorzi is with the Dept. of Information Engineering, University of Padova, Padova, Italy. email: zorzi@dei.unipd.it.}
\thanks{Part of this work was presented at IEEE ICC 2013 \cite{MichelusiICC12}.}
}

\maketitle
\begin{abstract}
Wireless sensors can integrate rechargeable batteries and energy-harvesting (EH) devices to enable long-term, autonomous operation, thus requiring intelligent energy management to limit the adverse impact of energy outages. This work considers a network of EH wireless sensors, which report packets 
with a random utility value
to a fusion center (FC) over a shared wireless channel. Decentralized access schemes are designed, where each node performs a \emph{local} decision to transmit/discard a packet, based on an estimate of the packet's 
utility, its own energy level, and the scenario state of the EH process, with the objective to maximize the average long-term aggregate
utility of the packets received at the FC. Due to the non-convex structure of the problem, 
an approximate optimization is developed by
resorting to a mathematical artifice based on a game theoretic formulation of the multiaccess scheme, where the nodes
do not behave strategically, but rather attempt to maximize a \emph{common} network utility
with respect to their own policy.
The \emph{symmetric Nash equilibrium} (SNE)
is characterized, where all nodes employ the same policy; its uniqueness is proved, and it is shown to be a local maximum of the original problem.
An algorithm to compute the SNE is presented, and a heuristic scheme is proposed, which is optimal for large battery capacity. 
It is shown numerically that the SNE typically achieves near-optimal performance, within 3\% of the optimal policy, at a fraction of the complexity, 
and two operational regimes of EH-networks are identified and analyzed: an \emph{energy-limited scenario}, where energy is scarce and the channel is under-utilized, and a \emph{network-limited scenario}, where energy is abundant and the shared wireless channel represents the bottleneck of the system.
 \end{abstract}
\section{Introduction}
\label{Sec:Intro}
In recent years, energy harvesting (EH) has received great interest in the research community,
 spanning different disciplines
related to electronics, energy storage, measurement and modeling of ambient energy sources \cite{Miozzo}, and energy management~\cite{Gunduz,Paradiso2005}.
EH devices can operate
autonomously over long periods of time, as they are capable
of collecting energy from the surrounding environment, and of employing it to sustain tasks such as data acquisition, processing and transmission.
Due to its demonstrated advantages of long-term, autonomous operation,
 at least within the limits of the physical degradation of the rechargeable batteries \cite{nicolo-exinfocom},
the EH technology exploiting, \emph{e.g.}, solar, motion, heat, aeolian
harvesting~\cite{Renner2010}, or indoor light \cite{Gorlatova2011}
is increasingly being considered in the design of wireless sensor networks (WSNs),
 where battery replacement is difficult or cost-prohibitive,
 \emph{e.g.}, see~\cite{Anthony2012}.

The rechargeable battery  of an EH sensor (EHS) can be modeled as an
\emph{energy buffer}, where energy is stored according to a
given statistical process, and from where energy is drawn
to feed sensor microprocessors and transceiver
equipments, whenever needed. 
However, due to the random and uncertain energy supply,
\emph{battery management} algorithms are needed to manage the harvested energy,
in order to minimize the  deleterious impact of energy outage, 
with the goal of optimizing the long-term performance of the overall system in regard to sensing and
data-communication tasks~\cite{Gunduz,Lin2005,Kar2006,Niyato2007,Kansal2007}.
In the literature, different battery management algorithms have been presented,
exploiting knowledge about, \emph{e.g.}, the amount of energy stored in the battery, which may be perfectly
\cite{Gatzianas2010,Sharma2010} or imperfectly \cite{MichelusiFusion} known to the controller,
the importance of the data packets to be transmitted \cite{nicolo-kostas},
the state of the EH source \cite{nicolo-kostas},
the health state of the battery \cite{nicolo-exinfocom}, or the channel state \cite{Ozel2011}.

In this paper, extending our previous work \cite{MichelusiICC12}, we consider an EH-WSN where multiple EHSs
randomly access a wireless channel to transmit data packets with random utility value
 to a FC.
 The goal is to design strategies so as to maximize the long-term average utility of the data reported to the FC.
One possibility to achieve this goal is to have the FC act as a centralized controller,
which schedules each EHS based on knowledge of the energy levels and packet utilities of all EHSs.
 This problem has been  considered for the case of two EHSs in \cite{DelTesta}.
 While this centralized scheme is expected to yield the best performance, since collisions between
 concurrent transmissions can be avoided via proper scheduling at the FC,
  it can only be used in very small WSNs, such as the one considered 
 in  \cite{DelTesta}. In fact,
the EHSs need to report their energy level and packet's utility to the FC in each slot,
thus  incurring a very large
 communication overhead in large WSNs, which is not practical given the limited resources (shared channel and ambient energy) of EH-WSNs.
Therefore, unlike \cite{DelTesta}, we consider a fully decentralized scheme with an arbitrary number of EHSs, where the EHSs  operate without the aid of a central controller.
In particular,
their decision as to whether to transmit the data packet or to remain idle and discard it is based only on \emph{local}
information about their own energy level and an estimate of the utility of the current data packet, rather than on \emph{global} network state information.
Compared to the centralized scheme, the decentralized one
achieves poorer performance, due to the unavoidable collisions between concurrent EHSs, but 
 is scalable to large EH-WSNs,
since it does not incur communication overhead between the EHSs and the FC, but only needs minimal coordination
by the FC, which, from time to time (\emph{i.e.}, whenever environmental conditions change), 
needs to broadcast information about the policy to be employed by the EHSs to perform their access decisions.
Assuming that data transmission incurs an energy cost
and that simultaneous transmissions from multiple EHSs cause collisions and packet losses,
we study the problem of designing optimal decentralized random access policies,
 which manage the energy available in the battery
so as to maximize the network utility, defined as the
average long-term aggregate network utility of the data packets
successfully reported to the FC. In order to tackle the non-convex nature of the optimization problem,
we resort to a mathematical artifice based on a game theoretic formulation of the multiaccess problem,
 where each EHS tries to individually maximize
the network utility.
Owing to the symmetry of the network,
 we characterize the symmetric Nash equilibrium of this game, such that all EHSs employ the same scheme to access the channel, and
prove its existence and uniqueness. We also show that it is a local optimum
of the original optimization problem, and derive an algorithm to compute it, based on the bisection method and policy iteration.
Moreover, we design a heuristic scheme, which is proved to be asymptotically optimal for large battery capacity, based on which
we identify two regimes of operation of EH-WSNs: an \emph{energy-limited scenario}, where the bottleneck of the system is due
to the scarce energy supply, resulting in under-utilization of the shared wireless channel; and a \emph{network-limited scenario},
where energy is abundant and the network size and the shared wireless channel are the main bottlenecks of the system.
  
 The model considered in the present work is a generalization of that of~\cite{Michelusi2012},
which addresses the design of optimal energy management policies for a single EHS.
Herein, we extend~\cite{Michelusi2012} and model the interaction among multiple EHSs,
which randomly access the channel.
The problem of maximizing the average long-term importance (similar to the concept of utility used in this work) of the data reported by
 a replenishable sensor is formulated in~\cite{Greenstein2005} for a continuous-time
model with Poisson EH and data processes, whereas \cite{Valles2011} investigates the relaying of packets of
 different priorities in a network of energy-limited sensors (no EH).
   
Despite the intense research efforts in the design of optimal energy management policies
for a single EH device, \emph{e.g.}, see~\cite{Jaggi2009,Sharma2010,Seyedi2010},
the problem of analyzing and modeling the interaction among multiple EH nodes in a network
 has received limited attention so far.
 In \cite{IannelloTCOM,IannelloCISS},
the design of medium access control (MAC) protocols for EH-WSNs is considered,
focusing on TDMA, (dynamic) framed Aloha, CSMA and polling techniques.
In \cite{Sharma08},
efficient energy management policies that 
stabilize the data queues and maximize throughput or minimize delay are considered, as well as efficient policies for contention free and contention based MAC schemes,
assuming that the data buffer and the battery storage capacity are infinite.
In contrast, we develop decentralized random multiaccess schemes under the practical assumption of \emph{finite} battery storage capacity.
In \cite{ZhiAngEu},
a MAC protocol is designed using a probabilistic polling mechanism that adapts to changing EH rates or node densities.
In \cite{Fafoutis},
an on-demand MAC protocol is developed, able to support individual duty cycles for nodes with different energy profiles, so as to achieve energy neutral operation.
In \cite{Kar2006},  a dynamic sensor activation framework to optimize the event detection performance of the network is developed.
In \cite{LongbiLin},  the problem of maximizing the lifetime of a network with finite energy capacity is considered, and a
heuristic scheme is constructed that achieves close-to-maximum lifetime, for battery-powered nodes without EH capability.
In \cite{Chen}, the problem of
joint energy allocation and routing in multihop EH networks in order to maximize the total system utility is considered,
and a low-complexity online scheme is developed that achieves asymptotic optimality, as the battery storage capacity becomes asymptotically large,
 without prior knowledge of the EH profile.
In \cite{Gatzianas2010,Tapparello,Huang}, tools from Lyapunov optimization and weight perturbation are employed
to design near-optimal schemes to maximize the network utility in multihop EH networks. 
In particular,  \cite{Gatzianas2010}  considers the problem of cross-layer resource allocation for wireless multihop networks operating with rechargeable batteries under general arrival, channel state and recharge processes, and presents an online and adaptive policy for the stabilization and optimal control of EH networks. 
In \cite{Tapparello}, the problem of dynamic joint compression and transmission in EH multihop networks with correlated sources is considered. 
In \cite{Huang}, the problem of joint energy allocation and data routing is considered, and near-optimal online power allocation schemes are developed.
However, \cite{Huang} does not consider wireless interference explicitly,
and is thus applicable
to settings where adjacent nodes operate on orthogonal channels. 
In contrast, we employ a collision channel model, where the simultaneous transmission of multiple nodes results in packet losses.

The rest of the paper is organized as follows.
In Sec.~\ref{Sec:sysmo}, we present the system model.
 Sec.~\ref{Sec:optprob} defines the control policies and states the optimization problem,
which is further developed in Sec.~\ref{Sec:analysis}.
In Sec.~\ref{heur}, we present and analyze the performance of a heuristic policy.
In Sec.~\ref{Sec:numer}, we present some numerical results.
Finally,  in Sec.~\ref{Sec:concl} we conclude the paper.
The proofs of the propositions and theorems are provided in the Appendix.

\section{System model}\label{Sec:sysmo}

\begin{table*}[t]
\caption{Main System Parameters}
\vspace{-5mm}
\label{tab}
\begin{center}
\footnotesize
\scalebox{0.88}{
\begin{tabular}{|c| l | c | l |}
\hline\T\B
 $e_{\max}$& Battery storage capacity & $E_{u,k}\in\mathcal E$& Energy level of EHS $u$, time $k$
\\\hline\T\B
 $B_{u,k}\sim\mathcal B(\beta(S_k))$& EH arrival at EHS $u$, time $k$ & $Q_{u,k}\in\{0,1\}$& Transmit action at EHS $u$, time $k$
\\\hline\T\B
$S_k\sim p_S(\cdot |S_{k-1})$& Common scenario process &
$\beta(s)$& Average EH rate in scenario $s$\\
\hline\T\B $V_{u,k}\sim f_{V}(\cdot)$& Packet's utility at EHS $u$, time $k$
&
$Y_{u,k}\sim f_{Y|V}(\cdot |V_{u,k})$& Packet's utility observation at EHS $u$, time $k$
 \\
\hline\T\B
$\gamma_{u,k}$ & Channel SNR for EHS $u$, time $k$ & $\bar V(Y_{u,k})$ & Expected utility for EHS $u$, if it transmits at time $k$
 \\
\hline\T\B $\eta_u(E_{u,k},S_k)$& Transmission probability for EHS $u$
&
$g(\eta_u(E_{u,k},S_k))$& Expected instantaneous utility under threshold policies
\\
\hline\T\B 
$R_{\bs\eta}$& Network utility
&
$P(\eta_u|s)$
&
Average power consumption in scenario $s$ for EHS $u$
\\
\hline\T\B 
$y_{\mr{th},u}(e,s)$
&
Censoring threshold on $Y_{u,k}$
&
$\rho(y)$ & Channel outage probability, given $Y_{u,k}=y$
\\
\hline\T\B 
$G(\eta_u|s)$
&
\multicolumn{3}{l|}{Average utility in scenario $s$ for EHS $u$ under no collisions}
\\
\hline
\end{tabular}
}
\vspace{-5mm}
\end{center}
\end{table*}

We consider a network of $U$ EHSs, which communicate concurrently 
via a shared wireless
link to a FC, as depicted in Fig.~\ref{fig:EHSnet}.
Each EHS harvests energy from the environment, and stores it in a 
rechargeable battery to power the sensing apparatus and the RF circuitry.
 A processing unit, \emph{e.g.}, a micro-controller, manages the energy consumption of the EHS.
Time is slotted, where slot $k$ is the time interval $[kT,kT+T)$, 
$k \in \mbb{Z}^+$ and $T$ is the time-slot duration.

\subsection{Data acquisition model}
At each time instant $k$,
each EHS (say, $u$) has a new data packet to send to the FC.
 Each data packet has utility $V_{u,k}>0$.
We model $V_{u,k}$ as a continuous random variable with probability density function (pdf) $f_{V}(v),\ v\geq 0$
with support $(0,\infty)$, and assume that the components of  $\{\mathbf V_{k}\}$,
where $\mathbf V_k=(V_{1,k},V_{2,k},\dots,V_{U,k})\in[\mathbb R^+]^U$ 
is the \emph{utility} vector,
 are i.i.d. over time and across the EHSs.
 EHS $u$ may not know exactly 
  the utility $V_{u,k}$ of the current data packet,
  but is only provided with a noisy observation $Y_{u,k}$ of it, termed the \emph{utility observation},
  distributed according to the conditional probability density function $f_{Y|V}(y|v)$ with support $(0,\infty),\ \forall v$.
 One example is when the utility $V_{u,k}$ represents the achievable information rate under the current channel fading state. In this case, 
$V_{u,k}$ is known only to some extent, via proper channel estimation.
 We denote the joint probability density function of $(V_{u,k},Y_{u,k})$ as $f_{V,Y}(v,y)$,
 the conditional probability density function of $V_{u,k}$ given $Y_{u,k}$ as $f_{V|Y}(v|y)$,
 and the marginal probability density function of $Y_{u,k}$ as $f_{Y}(y)$.

 \begin{figure}
    \centering
\scalebox{0.6}{
\begin{tikzpicture}
\draw [ultra thick, ->] (3,0) -- (0+0.5,0);
\draw [ultra thick, ->] (-2.12,2.12) -- (0-0.3535,0+0.3535);
\draw [ultra thick, ->] (0,-3) -- (0,0-0.5);
\draw [fill=black,draw=black,thick,text=white] (0,0) circle [radius=0.5];
\node [text=white] at (0,0) {FC};
\draw [fill=gray,draw=black,thick,text=white] (3,0) circle [radius=0.5];
\node at (3,0) {EHS1};
\node at (4,0) {$V_{1,k}$};
\draw [ultra thick, ->] (3.7,0) -- (3.5,0);
\def\x{3.5}
\def\y{0.35}
\draw [fill=gray,thick] (0+\x,0+\y) rectangle (0.2+\x,0.5+\y);
\draw [fill=gray,thick] (0.2+\x,0+\y) rectangle (0.4+\x,0.5+\y);
\draw [fill=gray,thick] (0.4+\x,0+\y) rectangle (0.6+\x,0.5+\y);
\draw [fill=white,thick] (0.6+\x,0+\y) rectangle (0.8+\x,0.5+\y);
\draw [fill=white,thick] (0.8+\x,0+\y) rectangle (1+\x,0.5+\y);
\draw [fill=white,thick] (1+\x,0+\y) rectangle (1.2+\x,0.5+\y);
\draw [ultra thick, ->] (\x+1.4,\y+0.25) -- (\x+1.2,\y+0.25);
\node at (\x+1.7,\y+0.2) {$B_{1,k}$};
\draw [fill=white, ultra thick] (2.12,2.12) circle [radius=0.5];
\node at (2.12,2.12) {EHS2};
\node at (2.83+0.2,2.83-0.2) {$V_{2,k}$};
\draw [ultra thick, ->] (2.62,2.62) -- (2.47,2.47);
\def\x{2.12+0.6}
\def\y{2.12-0.25}
\draw [fill=gray,thick] (0+\x,0+\y) rectangle (0.2+\x,0.5+\y);
\draw [fill=gray,thick] (0.2+\x,0+\y) rectangle (0.4+\x,0.5+\y);
\draw [fill=white,thick] (0.4+\x,0+\y) rectangle (0.6+\x,0.5+\y);
\draw [fill=white,thick] (0.6+\x,0+\y) rectangle (0.8+\x,0.5+\y);
\draw [fill=white,thick] (0.8+\x,0+\y) rectangle (1+\x,0.5+\y);
\draw [fill=white,thick] (1+\x,0+\y) rectangle (1.2+\x,0.5+\y);
\draw [ultra thick, ->] (\x+1.4,\y+0.25) -- (\x+1.2,\y+0.25);
\node at (\x+1.7,\y+0.2) {$B_{2,k}$};
\draw [fill=white, ultra thick] (0,3) circle [radius=0.5];
\node at (0,3) {EHS3};
\node at (0,4) {$V_{3,k}$};
\draw [ultra thick, ->] (0,3.7) -- (0,3.5);
\def\x{0.6}
\def\y{3-0.25}
\draw [fill=gray,thick] (0+\x,0+\y) rectangle (0.2+\x,0.5+\y);
\draw [fill=gray,thick] (0.2+\x,0+\y) rectangle (0.4+\x,0.5+\y);
\draw [fill=gray,thick] (0.4+\x,0+\y) rectangle (0.6+\x,0.5+\y);
\draw [fill=gray,thick] (0.6+\x,0+\y) rectangle (0.8+\x,0.5+\y);
\draw [fill=gray,thick] (0.8+\x,0+\y) rectangle (1+\x,0.5+\y);
\draw [fill=white,thick] (1+\x,0+\y) rectangle (1.2+\x,0.5+\y);
\draw [ultra thick, ->] (\x+1.4,\y+0.25) -- (\x+1.2,\y+0.25);
\node at (\x+1.7,\y+0.2) {$B_{3,k}$};
\draw [fill=gray,draw=black,thick,text=white] (-2.12,2.12) circle [radius=0.5];
\node at (-2.12,2.12) {EHS4};
\node at (-2.83,2.83) {$V_{4,k}$};
\draw [ultra thick, ->] (-2.62,2.62) -- (-2.47,2.47);
\def\x{-2.12-1.2-0.6}
\def\y{2.12-0.25}
\draw [fill=gray,thick] (0+\x,0+\y) rectangle (0.2+\x,0.5+\y);
\draw [fill=gray,thick] (0.2+\x,0+\y) rectangle (0.4+\x,0.5+\y);
\draw [fill=gray,thick] (0.4+\x,0+\y) rectangle (0.6+\x,0.5+\y);
\draw [fill=gray,thick] (0.6+\x,0+\y) rectangle (0.8+\x,0.5+\y);
\draw [fill=gray,thick] (0.8+\x,0+\y) rectangle (1+\x,0.5+\y);
\draw [fill=gray,thick] (1+\x,0+\y) rectangle (1.2+\x,0.5+\y);
\draw [ultra thick, ->] (\x-0.2,\y+0.25) -- (\x+0,\y+0.25);
\node at (\x-0.6,\y+0.2) {$B_{4,k}$};
\draw [fill=white, ultra thick] (-3,0) circle [radius=0.5];
\node at (-3,0) {EHS5};
\node at (-4,0) {$V_{5,k}$};
\draw [ultra thick, ->] (-3.7,0) -- (-3.5,0);
\def\x{-3.5-1.2}
\def\y{0.35}
\draw [fill=white,thick] (0+\x,0+\y) rectangle (0.2+\x,0.5+\y);
\draw [fill=white,thick] (0.2+\x,0+\y) rectangle (0.4+\x,0.5+\y);
\draw [fill=white,thick] (0.4+\x,0+\y) rectangle (0.6+\x,0.5+\y);
\draw [fill=white,thick] (0.6+\x,0+\y) rectangle (0.8+\x,0.5+\y);
\draw [fill=white,thick] (0.8+\x,0+\y) rectangle (1+\x,0.5+\y);
\draw [fill=white,thick] (1+\x,0+\y) rectangle (1.2+\x,0.5+\y);
\draw [ultra thick, ->] (\x-0.2,\y+0.25) -- (\x+0,\y+0.25);
\node at (\x-0.6,\y+0.2) {$B_{5,k}$};
\draw [fill=white, ultra thick] (-2.12,-2.12) circle [radius=0.5];
\node at (-2.12,-2.12) {EHS6};
\node at (-2.83-0.2,-2.83+0.2) {$V_{6,k}$};
\draw [ultra thick, ->] (-2.62,-2.62) -- (-2.47,-2.47);
\def\x{-2.12-1.2-0.6}
\def\y{-2.12-0.25}
\draw [fill=white,thick] (0+\x,0+\y) rectangle (0.2+\x,0.5+\y);
\draw [fill=white,thick] (0.2+\x,0+\y) rectangle (0.4+\x,0.5+\y);
\draw [fill=gray,thick] (0.4+\x,0+\y) rectangle (0.6+\x,0.5+\y);
\draw [fill=gray,thick] (0.6+\x,0+\y) rectangle (0.8+\x,0.5+\y);
\draw [fill=gray,thick] (0.8+\x,0+\y) rectangle (1+\x,0.5+\y);
\draw [fill=gray,thick] (1+\x,0+\y) rectangle (1.2+\x,0.5+\y);
\draw [ultra thick, ->] (\x-0.2,\y+0.25) -- (\x+0,\y+0.25);
\node at (\x-0.6,\y+0.2) {$B_{6,k}$};
\draw [fill=gray,draw=black,thick,text=white] (0,-3) circle [radius=0.5];
\node at (0,-3) {EHS7};
\node at (0,-4) {$V_{7,k}$};
\draw [ultra thick, ->] (0,-3.7) -- (0,-3.5);
\def\x{0-1.2-0.6}
\def\y{-3-0.25}
\draw [fill=white,thick] (0+\x,0+\y) rectangle (0.2+\x,0.5+\y);
\draw [fill=white,thick] (0.2+\x,0+\y) rectangle (0.4+\x,0.5+\y);
\draw [fill=white,thick] (0.4+\x,0+\y) rectangle (0.6+\x,0.5+\y);
\draw [fill=white,thick] (0.6+\x,0+\y) rectangle (0.8+\x,0.5+\y);
\draw [fill=gray,thick] (0.8+\x,0+\y) rectangle (1+\x,0.5+\y);
\draw [fill=gray,thick] (1+\x,0+\y) rectangle (1.2+\x,0.5+\y);
\draw [ultra thick, ->] (\x-0.2,\y+0.25) -- (\x+0,\y+0.25);
\node at (\x-0.6,\y+0.2) {$B_{7,k}$};
\draw [fill=white, ultra thick] (2.12,-2.12) circle [radius=0.5];
\node at (2.12,-2.12) {EHS8};
\node at (2.83,-2.83) {$V_{8,k}$};
\draw [ultra thick, ->] (2.62,-2.62) -- (2.47,-2.47);
\def\x{2.12+0.6}
\def\y{-2.12-0.25}
\draw [fill=white,thick] (0+\x,0+\y) rectangle (0.2+\x,0.5+\y);
\draw [fill=white,thick] (0.2+\x,0+\y) rectangle (0.4+\x,0.5+\y);
\draw [fill=white,thick] (0.4+\x,0+\y) rectangle (0.6+\x,0.5+\y);
\draw [fill=white,thick] (0.6+\x,0+\y) rectangle (0.8+\x,0.5+\y);
\draw [fill=white,thick] (0.8+\x,0+\y) rectangle (1+\x,0.5+\y);
\draw [fill=white,thick] (1+\x,0+\y) rectangle (1.2+\x,0.5+\y);
\draw [ultra thick, ->] (\x+1.4,\y+0.25) -- (\x+1.2,\y+0.25);
\node at (\x+1.7,\y+0.2) {$B_{8,k}$};
\end{tikzpicture}}
\caption{Energy Harvesting Wireless Sensor Network (EH-WSN)}
\label{fig:EHSnet}
\end{figure}
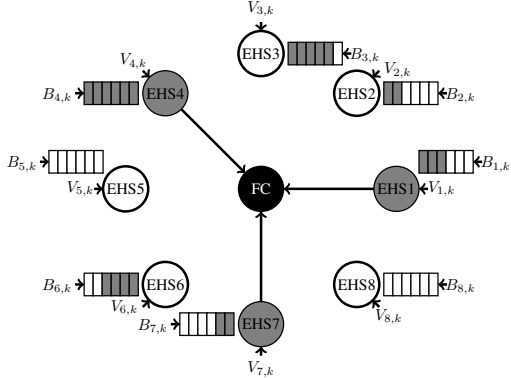

\subsection{Channel model}\label{chmodel}
We employ a collision channel model with fading,\footnote{The extension to the case of multiple orthogonal channels can be done
by exploiting the large network approximation, similar to \cite[Corollary 1]{MicheTSPp1}} \emph{i.e.},
letting $\gamma_{u,k}$ be the SNR between EHS $u$ and the FC in slot $k$,
the
transmission of EHS $u$ in slot $k$
is successful if and only if all the other EHSs remain 
idle and $\gamma_{u,k}>\gamma_{\mathrm{th}}(Y_{u,k})$, where $\gamma_{\mathrm{th}}(Y_{u,k})$
is an SNR threshold, possibly dependent on the utility observation $Y_{u,k}$, which guarantees correct decoding at the receiver.
The dependence on $Y_{u,k}$ models the fact that the transmission rate may be adapted over time, based on $Y_{u,k}$.
We model $\{\gamma_{u,k}\}$ as a continuous random variable in $(0,\infty)$,
  i.i.d. over time and across EHSs, and possibly correlated with $(V_{u,k},Y_{u,k})$.
For instance, if $V_{u,k}$ represents the achievable information rate under the
current channel fading state, then
$V_{u,k}=\log_2(1+\gamma_{u,k})$ \cite{Cover}, and $Y_{u,k}$ is a noisy estimate of the achievable rate, obtained from a noisy estimate of the SNR $\gamma_{u,k}$.
We denote the joint probability distribution of $(V_{u,k},Y_{u,k},\gamma_{u,k})$ as 
$f_{V,Y,\Gamma}(v,y,\gamma)$, and 
we let
$\rho(y)\triangleq\mathbb P(\gamma_{u,k}\leq\gamma_{\mathrm{th}}(y)|Y_{u,k}=y)$ be the channel outage probability, which depends on the specific fading model
and rate adaptation employed. 
Thus, letting $O_{u,k}\in\{0,1\}$ be the transmission outcome for EHS $u$, with
$O_{u,k}=1$ if the transmission is successful and $O_{u,k}=0$ otherwise, 
and letting $Q_{n,k}\in\{0,1\}$  be the transmission decision of EHS $n$, with
$Q_{n,k}=1$ if it transmits and $Q_{n,k}=0$ if it remains idle,
the success probability under a 
 specific transmission pattern $\mathbf Q_k=(Q_{1,k},Q_{2,k},\dots,Q_{U,k})\in\{0,1\}^U$ and $Y_{u,k}=y$ is given by
\begin{align}
\mathbb P(O_{u,k}{=}1|\mathbf Q_k{=}\mathbf q,Y_{u,k}{=}y){=}
(1{-}\rho(y))q_{u}\prod_{i\neq u}(1{-}q_{i}).
\end{align}
The data packet is discarded if a collision occurs or the EHS decides to remain idle,
and no feedback on the transmission outcome is provided by the FC to the EHSs.
This assumption implicitly models a stringent delay requirement.
Moreover, we do not deal explicitly with network management issues, \emph{e.g.}, node synchronization and self organization,
nodes joining and leaving the network, fault tolerance,
which can be addressed using techniques available in the literature.

Given the utility observation $Y_{u,k}=y$, EHS $u$ can estimate
the utility of the current data packet via conditional expectation, \emph{i.e.},
 \begin{align}
 \label{condexp}
 \bar V(y)\triangleq \mathbb E\left[\chi(\gamma_{u,k}>\gamma_{\mathrm{th}}(Y_{u,k}))V_{u,k}|Y_{u,k}=y\right].
\end{align}
Note that $\bar V(y)$ represents the expected utility for EHS $u$, if the packet is transmitted, and takes into account 
possible channel outage events due to channel fading. In the special case where $\gamma_{u,k}$ and $V_{u,k}$ are mutually independent, given $Y_{u,k}$,
 we obtain  $ \bar V(y)=(1-\rho(y))\mathbb E\left[V_{u,k}|Y_{u,k}=y\right]$.
 We make the following assumptions on $\bar V(y)$.
 \begin{ass}\label{assu}
 $\bar V(y)$ is a strictly increasing, continuous function of $y\in (0,\infty)$, with $\underset{y\to\infty}\lim \bar V(y)=~\infty$
 and $\bar V(0)=0$.
 \end{ass}
\begin{remark}
\label{appdp}
The joint distribution of $(V_{u,k},Y_{u,k},\gamma_{u,k})$ is application/deployment dependent.
Moreover, the assumption that  $V_{u,k}$, $Y_{u,k}$ and $\gamma_{u,k}$  have continuous distributions with support $(0,\infty)$
approximates the scenario where
they are mixed or discrete random variables as well. This is the case, for instance, if the data traffic is bursty, \emph{i.e.}, a data packet is received with probability $\lambda$,
otherwise no data packet is received.
In this case, the utility is binary and takes value $V_{u,k}=1$ if a packet is received and $V_{u,k}=0$ otherwise, whereas
the utility observation is given by $Y_{u,k}=V_{u,k}$.
The approximation for $V_{u,k}$ is obtained as follows:
let $V_{u,k}$ be a mixed random variable, \emph{i.e.}, with probability $x$, it takes a value from the discrete set $\mathcal V_{D}$ with probability
mass function $f_{V}^{(D)}(v),\ v\in\mathcal V_D$,
and, with probability $1-x$, it takes a value from the continuous set $\mathcal V_C$ with probability density function $f_{V}^{(C)}(v)$ (in the discrete case, $x=1$).
Then, $V_{u,k}$ can be approximated as a continuous random variable with support $(0,\infty)$ and probability density function
\begin{align*}
\tilde f_V^{\sigma}(v){=}(1{-}x)f_{V}^{(C)}(v){+}x
\!\!\!\!\!\sum_{\nu_D\in\mathcal V_D}\!\!\!f_{V}^{(D)}(\nu_D)\frac{\frac{1}{\sqrt{2\pi}\sigma}e^{-\frac{(v-\nu_D)^2}{2\sigma^2}}}{1{-}\Phi(\nu_D/\sigma)
},v{\geq}0,
\end{align*}
where $\Phi(x)=\mathbb P(N<x)$ is the cumulative distribution function of a Gaussian normal random variable $N\sim\mathcal N(0,1)$.
This approximation has support $(0,\infty)$, and approaches the true mixed distribution as $\sigma\to 0$.
The same consideration holds for the case where the  observation $Y_{u,k}$
and the fading state $\gamma_{u,k}$ have a discrete or mixed distribution.
Therefore, without loss of generality,
 we restrict our analysis to the case where $V_{u,k} $ and $Y_{u,k}$ are continuous random variables with support $(0,\infty)$.
\end{remark}

\subsection{Battery dynamics}
The battery of each EHS is modeled by a buffer.
 As in previous works \cite{Jaggi2009,Seyedi2010,Michelusi2012}, we assume that 
each position in the buffer can hold one energy quantum and that the transmission of one data packet requires the 
expenditure of one energy quantum.
 In this work, 
we assume that  no power adaptation is employed at the EHSs (however, the transmission rate may be adapted to the utility observation $Y_{u,k}$,
see Sec.~\ref{chmodel}). Moreover,
 we only consider the energy expenditure associated with
RF transmission, and neglect the one due to battery leakage effects and storage inefficiencies,
 or needed to turn on-off the circuitry, etc. These aspects
can be handled by extending the EH/consumption model, as discussed in \cite{nicolo-exinfocom},
but their inclusion does not bring additional insights on the multiaccess problem considered in this paper, and therefore is not pursued here.
The maximum number of quanta that can be stored in each EHS, \emph{i.e.}, the battery capacity,
is $e_{\mr{max}}$ and the set of possible energy levels is denoted 
by $\mc{E}{=}\{0,1,\dots,e_{\mr{max}}\}$. 
 The amount of energy stored in the battery of EHS $u$ at time $k$ is denoted
as $E_{u,k}$, and evolves as
\begin{align}
\label{eq:energyQueueEvol}
E_{u,k+1} &= \min\left\{E_{u,k}-Q_{u,k}+B_{u,k},e_{\mr{max}}\right\},
\end{align}
where $\{B_{u,k}\}$ is the {\em energy arrival process}
 and $\{Q_{u,k}\}$ is the
 {\em transmit action process} at EHS $u$.
 $Q_{u,k}=1$ if the current data packet is
transmitted by EHS $u$, which results in the expenditure of one energy quantum,
 and $Q_{u,k}=0$ otherwise, as described in Sec. \ref{Sec:optproblem}.
$B_{u,k}$ models the randomness in the energy harvested in slot $k$,
and is described in Sec. \ref{scproc}.
Moreover the energy harvested in time-slot $k$ can be used only in a later time-slot,
so that, if the battery is depleted, \emph{i.e.}, $E_{u,k}=0$, then $Q_{u,k}=0$.
The state of the system at time $k$ is given by $(\mathbf E_k,\mathbf V_k)$, 
 where $\mathbf E_k = (E_{1,k},E_{2,k},\dots,E_{U,k})\in\mathcal E^U$
is the joint state of the energy levels in all EHSs.

\subsection{Energy harvesting process}\label{scproc}
We model the energy arrival process $\left\{B_{u,k}\right\}$ as a
homogeneous hidden Markov process, taking values in the set
$\mc{B}$. We define an underlying \emph{common scenario process}
$\left\{S_k\right\}$, taking values in the finite set $\mathcal S$, which evolves according to a stationary irreducible Markov
chain with transition probability
$p_S(s_{k+1}|s_k)\triangleq\mathbb P(S_{k+1}{=}s_{k+1}|S_k{=}s_k)$.
Given the common scenario $S_k=s$, the energy harvest at EHS $u$ is modeled as 
$B_{u,k}\sim\mathcal B(\beta(s))$,
{\em i.e.}, 
in each time-slot,
either one energy quantum is harvested with probability $\beta(s)$, or no energy is harvested at all.
The common scenario process models different environmental conditions affecting the EH mechanism, for instance, 
for a solar source,
 $S_k\in\{\text{direct sunlight},\text{cloudy},\text{night}\}$.
 In this case, we thus have that $\beta(\text{direct sunlight})> \beta(\text{cloudy})>\beta(\text{night})$.
We assume that the components of
$\mathbf B_k=(B_{1,k},B_{2,k},\dots,B_{U,k})$
are i.i.d. over time and across the EHSs, given the common scenario $S_k=s$.
We define $\pi_S(s),\ s\in\mathcal S$, as the steady state
distribution of the common scenario process, and we refer to $\bar\beta
= \mbb{E}[B_{u,k}]=\sum_{s\in\mathcal S}\pi_S(s)\beta(s)$ as the \emph{average EH rate},
and to $\beta(s)$ as  the  \emph{average EH rate in scenario $S_k=s$}.
This model is a special instance of the \emph{generalized Markov model} presented in \cite{Keong} for the case of one EHS.
Therein, the scenario process is modeled as a first-order Markov chain, whereas
 $B_{u,k}$ statistically depends on $(B_{u,k-L},B_{u,k-L+1},\dots,B_{u,k-1})$ and on $S_k$, for some order $L\geq 0$.
In particular, in \cite{Keong} it is shown that $L=0$ ($\{B_{u,k}\}$ is i.i.d., given $S_k$)
models well a piezo-electric energy source, whereas  $L=1$ ($\{B_{u,k}\}$ is a Markov chain, given $S_k$) models well a solar energy source.
In this paper, the model $L=0$ is considered. The analysis can be extended to the case $L=1$ by including $B_{k-1}$ in the state space, similarly to \cite{nicolo-exinfocom}, but this extension is beyond the scope of this paper.
Typically, the scenario process $\{S_k\}$ is slowly varying over time,
hence it can be estimated accurately at each EHS from the locally harvested sequence $\{B_{u,k}\}$. In this paper, we assume that perfect knowledge of $S_{k}$ is
available at each EHS at time instant $k$.
Note that the common scenario process $\{S_k\}$ introduces \emph{spatial and temporal correlation} in the EH process, \emph{i.e.}, the local EH processes
are spatially coupled across the network since $S_k$ is common to all EHSs and time-correlated since $S_k$ is a Markov chain.

\section{Policy Definition and Optimization Problem}\label{Sec:optprob}
\label{Sec:optproblem}
Each EHS is assumed to have only \emph{local} knowledge
about the state of the system. Namely, at time $k$,
EHS $u$ only knows its own energy level $E_{u,k}$, the utility observation
$Y_{u,k}$, and the common scenario state $S_k$,
 but does not know the energy levels and utilities of the other EHSs in the network, nor the exact utility of its own packet $V_{u,k}$
 or channel SNR $\gamma_{u,k}$,
since these cannot be directly measured by EHS $u$.
Moreover, since  no feedback is provided by the FC, EHS $u$ has no information on the collision history. 
 Therefore, the transmission decision of EHS $u$ at time $k$ is based solely on
 $(E_{u,k},Y_{u,k},S_k)$. 
 
 Given $(E_{u,k},Y_{u,k},S_k)$,
 EHS $u$ transmits the current data packet ($Q_{u,k}=1$) with probability $\mu_u(E_{u,k},Y_{u,k},S_k)$, and discards it otherwise ($Q_{u,k}=0$),
 for some transmission policy $\mu_u$, which is the objective of our design.
 Note that, due to the complexity of implementing a non-stationary policy $\mu_{u,k}$, we only consider
 stationary policies independent of time $k$.
 Let  $\bs{\mu}=(\mu_1,\mu_2,\dots,\mu_U)$ be the joint transmission policy of the network.
Given an initial state of the energy levels $\mathbf E_0=\mathbf e_0\in\mathcal E^U$
and initial scenario $S_0=s_0\in\mathcal S$,
we denote the average long-term utility of the data reported by EHS $u$ to the FC,
under the joint transmission policy $\bs{\mu}$,  as
\begin{align}\label{Gu}
R_{\bs\mu}^{(u)}(\mathbf e_0,s_0){=}&\underset{K \rightarrow \infty}{\lim\inf}\frac{1}{K}
\mbb{E}\left[\sum_{k=0}^{K-1}\chi(\gamma_{u,k}>\gamma_{\mathrm{th}}(Y_{u,k}))Q_{u,k}V_{u,k}
\right.\nonumber\\&
\left.\left.
\times
\prod_{i\neq u}(1{-}Q_{i,k})\right|\mathbf E_0=\mathbf e_0,S_0=s_0,\bs\mu\right],
\end{align}
where $E_{u,k}$ evolves as in (\ref{eq:energyQueueEvol}),
$B_{u,k}\sim\mathcal B(\beta(S_k))$, $\{S_k\}$ is a Markov chain,
and $\chi(\cdot)$ is the indicator function.
The expectation is taken with respect to 
$\{\mathbf B_{k},S_k,\mathbf Q_{k},\mathbf V_{k},\mathbf Y_{k},\boldsymbol{\gamma}_k\}$,
where  $\mathbf Y_k=(Y_{1,k},Y_{2,k},\dots,Y_{U,k})$, $\boldsymbol{\gamma}_k=(\gamma_{1,k},\gamma_{2,k},\dots,\gamma_{U,k})$,
 and, at each instant $k$,
 $Q_{u,k}{=}1$ with probability $\mu_u(E_{u,k},Y_{u,k},S_k)$ and  $Q_{u,k}=0$ otherwise.
We define
the \emph{network utility} as
 the average long-term aggregate utility of the packets successfully reported to the FC, \emph{i.e.},
\begin{equation}
 \label{eq:averageGain}
R_{\bs\mu}(\mathbf e_0,s_0) = \sum_{u=1}^{U}R_{\bs\mu}^{(u)}(\mathbf e_0,s_0).
\end{equation}
The goal is to design a joint transmission policy $\bs\mu$ so as to maximize the network utility, \emph{i.e.},
\begin{align}
 \bs\mu^*=\arg\max_{\bs\mu}R_{\bs\mu}(\mathbf e_0,s_0).
\end{align}
We now establish that $\bs\mu^*$ has a threshold structure with respect to the observation $Y_{u,k}$.
\begin{propos}
\label{lemma:thresStruct}
For each energy level $E_{u,k}\in\mc{E}$, scenario process $S_k\in\mathcal S$, and for each EHS $u\in\{1,2,\dots,U\}$, there exists a threshold on the
utility observation, $y^*_{\mr{th},u}(E_{u,k},S_k)$, such that
\begin{align}\label{muthresstruc}
\mu_u^*(E_{u,k},Y_{u,k},S_k)=\chi\left(Y_{u,k}\geq y^*_{\mr{th},u}(E_{u,k},S_k)\right).
\end{align}
\end{propos}
\begin{proof}
 See Appendix~A.
\end{proof}
From Prop. \ref{lemma:thresStruct}, it follows that $Q_{u,k}{=}\chi\left(Y_{u,k}{\geq}y_{\mr{th},u}(E_{u,k},S_k)\right)$ is optimal,
 where  $y_{\mr{th},u}(E_{u,k},S_k)$ is the \emph{censoring threshold}
on $Y_{u,k}$, and is a function of the energy level $E_{u,k}$ and scenario state~$S_k$.
  
We denote the
 transmission probability of EHS $u$ in energy level $e$ and scenario $s$,
 induced by the random observation $Y_{u,k}$ and by the threshold $y_{\mr{th},u}(e,s)$,  as
\begin{equation}\label{eta}
 \eta_u(e,s){=}\mathbb E[Q_{u,k}|E_{u,k}{=}e,S_k{=}s]{=}\mathbb P(Y_{u,k}{\geq}y_{\mr{th},u}(e,s)).
\end{equation}
We denote the corresponding threshold on the utility observation as $y_{\mr{th}}(\eta_u(e,s))\triangleq y_{\mr{th},u}(e,s)$.
Moreover, we denote
the expected utility reported by EHS $u$ to FC in state $(e,s)$,
assuming that all the other EHSs remain idle (no collisions occur),
as $g(\eta_u(e,s))$. This is given~by
\begin{align}\label{gx}
\nonumber
g(\eta_u(e,s))&{=}\mathbb E\left[\chi(\gamma_{u,k}{>}\gamma_{\mathrm{th}}(Y_{u,k}))Q_{u,k}V_{u,k}|E_{u,k}{=}e,S_k{=}s\right]
\\&
=
\int_{y_{\mr{th},u}(e,s)}^{\infty}\bar V(y)f_{Y}(y)\mr{d}y.
\end{align}
In words, $g(\eta_u(e,s))$ 
is the expected reward accrued when only the packets with utility observation
 larger than $y_{\mr{th},u}(e,s)$ are reported
and no collisions occur (however, poor channel fading may still cause the transmission to fail), where $\eta_u(e,s)$ is the corresponding transmission probability.
 The function $g(x)$ has the following properties, which are not explicitly proved here due to space constraints.
\begin{propos}
\label{lemma:gx}
The function $g(x)$ is a continuous, strictly increasing, strictly concave function of $x$, with
$g(0){=}0$, $g(1){=}\mathbb E[\bar V(Y_{u,k})]$,
$g^\prime(x){\triangleq}\frac{\mathrm{d}g(x)}{\mathrm dx}{=}\bar V(y_{\mr{th}}(x))$,
 $\lim_{x\to 0}g^\prime(x){=}\infty$,  $\lim_{x\to 1}g^\prime(x){=}0$.
\end{propos}
Note that, from Assumption \ref{assu}, there is a one-to-one mapping between the threshold policy $\mu_u(e,y,s)$, the threshold $y_{\mathrm{th},u}(e,s)$,
and the transmission probability $\eta_u(e,s)$.
Therefore, without loss of generality, in the following we refer to $\eta_u$ as the policy of EHS $u$.
Moreover, we denote the aggregate policy used by all the EHSs in the network
as $\bs{\eta}=(\eta_1,\eta_2,\dots,\eta_U)$.
The utility for EHS $u$ under a threshold policy $\bs{\eta}$ is then given by
\begin{align}
 \label{eq:averageGainTH}
&R_{\bs\eta}^{(u)}(\mathbf e_0,s_0) {=} 
\underset{K \rightarrow \infty}{\lim\inf}\frac{1}{K}
\mbb{E}\left[ \sum_{k=0}^{K-1}g(\eta_u(E_{u,k},S_k))
\right.
\nonumber\\&
\left.\left.
\times
\prod_{i\neq u}(1-\eta_i(E_{i,k},S_k))\right|\mathbf E_0=\mathbf e_0,S_k=s_0,\bs\eta\right],
\end{align}
where $E_{u,k}$ evolves as in (\ref{eq:energyQueueEvol}),
$B_{u,k}\sim\mathcal B(\beta(S_k))$ and $\{S_k\}$ is a Markov chain.
The optimization problem under threshold policies  can be restated as
\begin{align}\label{optor}
 \bs\eta^*=\arg\max_{\bs\eta}R_{\bs\eta}(\mathbf e_0,s_0).
\end{align}
Note that the policy $\eta_u$ implemented by EHS $u$ yields the following trade-off.
By employing a larger transmission probability $\eta_u(E_{u,k},S_k)$:
\begin{itemize}
\item EHS $u$ transmits more often,  hence the 
battery level $E_{u,k}$ tends to be lower  since $Q_{u,k}\sim\mathcal B(\eta_u(E_{u,k},S_k))$ (see (\ref{eq:energyQueueEvol})) and
battery depletion occurs more frequently for EHS $u$, degrading its performance;
\item more frequent collisions occur in the channel, hence a lower utility is accrued by all the other EHSs,
as apparent from the term $\prod_{i\neq j}(1-\eta_i(E_{i,k},S_k))$ in (\ref{eq:averageGainTH}),
 denoting the probability of successful transmission for EHS $j$, which is a decreasing function of $\eta_u(E_{u,k},S_k)$;
\item a larger instantaneous utility is accrued by EHS $u$, as apparent from the term $g(\eta_u(E_{u,k},S_k))$ in (\ref{eq:averageGainTH}).
\end{itemize}
On the other hand, the opposite effects are obtained by using a smaller transmission probability.
The optimal policy $ \bs\eta^*$ thus reflects an optimal trade-off between
battery depletion probability, collision probability and instantaneous utility for each EHS.

We consider the following set
of \emph{admissible} policies $\mathcal U$.
\begin{definition}
\label{Def:admisspol}
The set $\mathcal U$ of admissible policies is defined as
\begin{align*}
\mathcal U=&\{\eta:\eta(0,s)=0,\eta(e_{\mr{max}},s)\in(0,1],\eta(e,s)\in(0,1),
\nonumber\\&
\forall e=1,2,\dots,e_{\max}-1,\ \forall s\in\mathcal S\}.
\end{align*}
\end{definition}
\begin{remark}
Note that considering only admissible policies $\eta\in\mathcal U$ does not incur a loss of optimality.
In fact, under a non-admissible policy, \emph{e.g.}, such that $\eta(e,s)=0$ for some $e>0$,
then the energy levels below ``$e$" are never visited (since the EHS never transmits in state $E_{u,k}=e$, hence the energy level cannot decrease).
It follows that the system is equivalent to another system with a smaller battery with capacity $\tilde e_{\max}=e_{\max}-e$.
Similarly, if $\eta(e,s)=1$ for some $e<e_{\max}$, then the energy levels above ``$e$" are never visited,
hence the system is equivalent to another system with a smaller battery with capacity $\tilde e_{\max}=e$.
In both cases, a non-admissible policy results in under-utilization of the battery storage capacity.
\end{remark}
It can be shown that the Markov chain $\{(\mathbf E_k,S_k)\}$ under the aggregate policy $\bs\eta\in\mathcal U^U$
is irreducible. Hence, there exists a unique steady-state distribution, $\pi_{\bs\eta}(\mathbf e,s), \mathbf e \in \mc{E}^U,s\in\mathcal S$,
independent of $(\mathbf{e}_0,s)$ \cite{DJWhite}.
From (\ref{eq:averageGainTH}), we thus obtain
\begin{align}
\label{eq:Gsteady}
R_{\bs\eta}^{(u)}{=}\!\!\!\!\sum_{\mathbf e\in\mathcal E^U,s\in\mathcal S}\pi_{\bs\eta}(\mathbf e,s)
g(\eta_u(e_u,s))\prod_{i\neq u}(1-\eta_i(e_i,s)).
\end{align}
Note that, given $S_k$,
the action $Q_{u,k}$ is based only on $(E_{u,k},Y_{u,k})$
and is independent of $(E_{i,k},Y_{i,k}),\ i\neq u$.
However, the transmission decisions over the network are coupled via the common scenario state $S_k$.
Herein, we use the following assumption:
\begin{ass}
The scenario process varies slowly over time, at a time scale much larger than battery dynamics.
\end{ass}
Therefore, conditioned on $S_k=s$,
a steady-state distribution of the energy levels is approached in each interval during which the scenario stays the same.
Then, using an approach similar to \cite{nicolo-exinfocom},
since the EH  process is i.i.d. across EHSs,
 the energy level of EHS $u$ is independent
of the energy levels of all the other EHSs,
so that we can approximate $\pi_{\bs\eta}(\mathbf e,s)$ as 
\begin{align}
\pi_{\bs\eta}(\mathbf e,s)=\pi_S(s)\prod_u\pi_{\eta_u}(e_u|s),
\end{align}
where $\pi_{\eta_u}(e_u|s)$ is the steady state probability of the energy level of EHS $u$ in scenario $S_k=s$, given by the following proposition.
\begin{propos}
\label{SSD}
Let $\eta{\in}\mathcal U$ and
 $\xi_\eta(e|s){\triangleq}\frac{\beta(s)(1-\eta(e,s))}{(1-\beta(s))\eta(e,s)}$; then,
 $\pi_{\eta}(e|s),\ e\in\mathcal E,s\in\mathcal S$ is given by
\begin{align}
\label{pi0}
&\pi_\eta(0|s)=\frac{1}{
1+\sum_{j=1}^{e_{\max}}\left[\prod_{f=1}^{j-1}\xi_\eta(f|s)\right]\frac{\beta(s)}{(1-\beta(s))\eta(j,s)}
},
\\&
\label{pie}
\pi_\eta(e|s)=\left[\prod_{f=1}^{e-1}\xi_\eta(f|s)\right]\frac{\beta(s)\pi_\eta(0|s)}{(1-\beta(s))\eta(e,s)},e>0.
\end{align}
\end{propos}
\begin{proof}
We apply the balance equations 
\begin{align}
&\pi_\eta(e|s)\beta(s)(1-\eta(e,s))=\pi_\eta(e+1|s)(1-\beta(s))\eta(e+1,s),
\nonumber\\&
\quad e=0,1,\dots,e_{\max}-1,
\end{align}
from which (\ref{pie}) is obtained by induction on $e$.
Finally, (\ref{pi0}) is obtained by normalization.
\end{proof}
Letting 
\begin{align}
\label{GP}
&G(\eta_u|s)=\sum_{e=1}^{e_{\max}}\pi_{\eta_u}(e|s)g(\eta_u(e,s)),\\\nonumber
& P(\eta_u|s)=\sum_{e=1}^{e_{\max}}\pi_{\eta_u}(e|s)\eta_u(e,s),
\end{align}
 we can rewrite (\ref{eq:Gsteady}) as
\begin{align}
\label{eq:Gsteady2}
R_{\bs\eta}^{(u)}=\sum_{s\in\mathcal S}\pi_S(s)G(\eta_u|s)\prod_{i\neq u}(1-P(\eta_i|s)).
\end{align}
Eq. (\ref{eq:Gsteady2}) can be interpreted as follows. 
$G(\eta_u|s)$ is the average long-term reward of EHS $u$ in scenario $s$, assuming that 
all the other EHSs remain idle, so that no collisions occur.
$P(\eta_i|s)$ is the average long-term transmission probability of EHS $i$ in scenario $s$,
so that $\prod_{i\neq u}(1-P(\eta_i|s))$ is the steady-state probability that  all EHSs $i\neq u$ remain idle,
 \emph{i.e.}, the transmission of EHS $u$ is successful given that it transmits.
From (\ref{eq:averageGain}), the network utility under the aggregate policy $\bs\eta=(\eta_1,\eta_2,\dots,\eta_U)$
then becomes
\begin{align}\label{Gagg}
R_{\bs\eta}=\sum_{s\in\mathcal S}\pi_S(s)
R_{\bs\eta}[s],
\end{align}
where we have defined the network utility in scenario $s$ as
\begin{align}\label{Gaggs}
R_{\bs\eta}[s]\triangleq
\sum_{u=1}^UG(\eta_u|s)\prod_{i\neq u}(1-P(\eta_i|s)).
\end{align}
\subsection{Symmetric control policies}
In order to guarantee fairness among the EHSs,
we consider only symmetric control policies, \emph{i.e.}, all the EHSs employ the same policy $\eta_u=\eta,\ \forall u$.
From (\ref{Gagg}) and (\ref{Gaggs}), we thus obtain
\begin{align}\label{rewtot}
R_{\bs\eta}=\sum_{s\in\mathcal S}\pi_S(s)
UG(\eta|s)(1-P(\eta|s))^{U-1}.
\end{align}
The optimization problem (\ref{optor}) over the class of admissible and symmetric policies is stated as
\begin{align}\label{origoptjoint}
\eta^* = \arg\max_{\eta\in\mathcal U}\sum_{s\in\mathcal S}\pi_S(s)
UG(\eta|s)(1-P(\eta|s))^{U-1}.
\end{align}
Moreover, since each term within the sum depends on the policy used in scenario $s$ only, and is independent of the policies used in all other scenarios, it follows that 
we can decouple the optimization as 
\begin{align}\label{origopt}
\eta^*(\cdot,s) = \underset{\eta(\cdot,s)\in\mathcal U}{\arg\max}
UG(\eta|s)(1-P(\eta|s))^{U-1}, \forall s\in\mathcal S.
\end{align}
Therefore, the optimal threshold policy can be determined independently for each scenario $s$, rather than jointly.
In the next section, we focus on the solution of the optimization problem (\ref{origopt}), for a generic scenario state $s\in\mathcal S$.

We have the following upper bound to $R_{\bs\eta}$.
\begin{propos}\label{upbound}
The network utility $R_{\bs\eta}$ under symmetric threshold policies is upper bounded by
\begin{align}
\label{ub2}
R_{\bs\eta}&{<}
\sum_{s\in\mathcal S}\pi_S(s)
Ug(\min\{x^*,\beta(s)\})(1-\min\{x^*,\beta(s)\})^{U-1}
\nonumber\\&
\triangleq R^{(up)},
\end{align}
where $x^*{\in}(0,1/U)$ uniquely solves $\bar V(y_{\mr{th}}(x^*))(1{-}x^*){=}(U{-}1)g(x^*)$.
\end{propos}
\begin{proof}
Using the strict concavity of $g(x)$ (Prop. \ref{lemma:gx}) and Jensen's inequality,
$G(\eta|s)$ is upper bounded by
$G(\eta|s)<g\left(\sum_{e=1}^{e_{\max}}\pi_{\eta}(e|s)\eta(e,s)\right)=g(P(\eta|s))$.
Moreover, the EH operation in scenario $s$ with EH rate $\beta(s)$ induces the constraint on the average transmission probability 
$P(\eta|s)\leq \beta(s)$.
Therefore, the network utility (\ref{rewtot}) can be upper bounded  by
\begin{align}\label{upbound2}
R_{\bs\eta}&<
\sum_{s\in\mathcal S}\pi_S(s)
Ug(P(\eta|s))(1-P(\eta|s))^{U-1}
\nonumber\\
&
\leq 
\sum_{s\in\mathcal S}\pi_S(s)
\max_{x\in[0,\beta(s)]} F(x),
\end{align}
where we have defined $F(x)\triangleq Ug(x)(1-x)^{U-1}$.
The first derivative of $F(x)$ with respect to $x$ is given by
\begin{align}
F^\prime(x)=U(1-x)^{U-2}[g^\prime(x)(1-x)-(U-1)g(x)].
\end{align}
Therefore $F(x)$ is an increasing function of $x$ if and only if $f(x)\triangleq g^\prime(x)(1-x)-(U-1)g(x)>0$.
Using the strict concavity of $g(x)$, we have that $f(x)$ is a strictly decreasing function of $x$, with
$\lim_{x\to 0}f(x)=\infty$ and
$f(1/U)=-[g(1/U)-\frac{1}{U}g^\prime(1/U)]<0$.
Therefore, there exists a unique $x^*\in (0,1/U)$ such that $f(x^*)=0$, $f(x)>0,\forall x\in (0,x^*)$
and $f(x)<0,\forall x\in (x^*,1/U)$.
Equivalently, $F(x)$ is a strictly increasing function of $x\in (0,x^*)$ and
 strictly decreasing function of $x\in (x^*,1/U)$, so that it is maximized by $x^*$.
Using Prop. \ref{lemma:gx} and the definition of $f(x)$, $x^*$ is the unique solution of  $\bar V(y_{\mathrm{th}}(x^*))(1-x^*)=(U-1)g(x^*)$.
 Then, by restricting the optimization of $F(x)$ over $x\in [0,\beta(s)]$ we obtain $R_{\bs\eta}\leq \sum_{s\in\mathcal S}\pi_S(s)F(\min\{x^*,\beta(s)\})$. The proposition is thus proved.
\end{proof}
In Sec.~\ref{heur}, we will present a heuristic policy, which will be shown to achieve the upper bound (\ref{upbound2}) for asymptotically large 
battery capacity $e_{\max}\to\infty$.
\section{Optimization and Analysis of Threshold policies}\label{Sec:analysis}
The optimization problem (\ref{origopt}) for the case of one EHS ($U=1$) has been studied in detail in \cite{Michelusi2012} and can be solved using the \emph{policy iteration algorithm} (PIA) \cite{Bertsekas2005}.
However, in general, when $U>1$, (\ref{origopt}) cannot be recast as a convex optimization problem,
hence we need to resort to approximate solutions.
In particular, in order to determine a local optimum of (\ref{origopt}),
 we use a mathematical artifice based on a game theoretic formulation of the multiaccess problem considered in this paper:
we model the optimization problem as a game,
 where it is assumed that each EHS, say $u$, is a player
which attempts to maximize the common payoff (\ref{Gaggs}) by optimizing its own policy $\eta_u$.
Note that the EHSs do not behave strategically, but rather attempt to maximize a \emph{common} network utility
with respect to their own policy.

We proceed as follows.
We first characterize the general Nash equilibrium (NE) of this game, thus defining the policy profile $\bs{\eta}^*$,
over the general space of non symmetric policies $\bs\eta\in\mathcal U^U$.
By definition of NE, $\bs{\eta}^*$ is the joint policy such that, 
if EHS $u$ deviates from the equilibrium solution by using a different policy $\eta_u\neq\eta_u^*$, while all other EHSs $i\neq u$ use the policy achieving the NE, $\eta_i^*$,
then a smaller network utility $R_{\bs\eta}$ is obtained. This equilibrium condition must hold
for all EHSs $u$.
Equivalently, any \emph{unilateral} deviation from the NE $\bs{\eta}^*$ (\emph{i.e.}, one and only one EHS deviates from it)
yields a smaller network utility $R_{\bs\eta}$.
Then, since our focus is on symmetric policies, we study the existence of a \emph{symmetric NE} (SNE) for this game,
\emph{i.e.}, a NE characterized by the symmetry condition $\eta_u^*=\eta^*,\ \forall u$, so that all EHSs employ the same policy.
In Theorem \ref{propeta}, we present structural properties of the SNE.
In Theorem \ref{thm1},
we show that the SNE is unique, and  we provide Algorithm~\ref{algoSNE} to compute it.
In Theorem \ref{thm2}, we prove that the SNE, and thus the policy 
returned by Algorithm~\ref{algoSNE}, is
 a local optimum of the original optimization problem (\ref{origopt}).
 Since in this section we consider the optimization problem (\ref{origopt}) for a specific scenario state $s\in\mathcal S$,
for notational convenience we neglect the dependence of $\eta$, $\pi_\eta$,
 $G(\eta_u|s)$ and $P(\eta_u|s)$
  on $s$, \emph{i.e.}, we write
 $\eta$, $\pi_\eta(e)$, $G(\eta_u)$ and $P(\eta_u)$ instead of $\eta(\cdot,s)$, $\pi_\eta(e|s)$, $G(\eta_u|s)$ and $P(\eta_u|s)$, respectively.
 
If a NE exists for this game (not necessarily symmetric), defined by the policy profile $\bs\eta^*=(\eta_1^*,\eta_2^*,\dots,\eta_U^*)$,
then any unilateral deviation from $\bs\eta^*$ yields a smaller network utility,
so that the NE must simultaneously solve, $\forall u$,
\begin{align}\label{NEor}
\begin{array}{rl}
\!\!\!\!\!\!\!\eta_u^*=\!\!\!\!&
\arg\max_{\eta_u\in\mathcal U}R_{\eta_u,\bs\eta_{-u}^*}[s]
\\=\!\!\!\!&
\arg\max_{\eta_u\in\mathcal U}
\left[G(\eta_u)\prod_{i\neq u}(1-P(\eta_i^*))\right.\\&\left.+(1-P(\eta_u))\sum_{n\neq u} G(\eta_n^*)\prod_{i\neq n,u}(1-P(\eta_i^*))
\right]
\\
=\!\!\!\!&
\arg\max_{\eta_u\in\mathcal U}
\left[G(\eta_u)-P(\eta_u)\sum_{n\neq u}\frac{G(\eta_n^*)}{1-P(\eta_n^*)}\right].
\end{array}\!\!
\end{align}
In the second step, using (\ref{Gaggs}), we have explicated the dependence of the network utility on the policy of EHS $u$, $\eta_u$
and, in the last step, we have
divided the argument of the optimization by the term $\prod_{i\neq u}(1-P(\eta_i^*))$, which is independent of $\eta_u$, and hence does not affect the optimization.

By further imposing a symmetric policy $\eta_u^*=\eta^*,\ \forall u,$ in (\ref{NEor}), we obtain
the SNE
\begin{align}\label{NE}
\eta^*=\arg\max_{\eta\in\mathcal U}\left[G(\eta)-\Lambda(\eta^*)P(\eta)\right],
\end{align}
where we have defined
\begin{align}\label{Lambda}
\Lambda(\eta)=(U-1)\frac{G(\eta)}{1-P(\eta)}.
\end{align}
Note that $\eta^*$ defined in (\ref{NE}) is \emph{simultaneously} optimal for all the EHSs,
\emph{i.e.}, by the definition of Nash equilibrium, any \emph{unilateral} deviation of a single EHS from the SNE $\eta^*$
yields a smaller network utility $R_{\bs\eta}$.\footnote{Note that these unilateral deviations do not preserve the symmetry of the policy, since, by definition, they assume that one and only one
EHS deviates from it.}
Moreover, although the SNE imposes symmetric policies, it represents an equilibrium over the larger space of non-symmetric policies, $\bs\eta\in\mathcal U^U$,
as dictated by the NE in (\ref{NEor}).
The interpretation of (\ref{NE}) is as follows. $G(\eta)$ is the reward obtained by EHS $u$ when 
 all the other EHSs are always idle. The term $\Lambda(\eta^*)$
can be interpreted as a Lagrange multiplier associated to a constraint on the transmission probability of
EHS $u$, so as to limit the collisions caused by its transmissions to the other EHSs in the network. The overall objective function in (\ref{NE}) is thus interpreted as the maximization 
of the individual reward of each node, with a constraint on the average transmission probability to limit the
collisions, which are deleterious to network performance. Interestingly, the Lagrange
multiplier (\ref{Lambda}) increases with the number of EHSs $U$, so that,
the larger the network size, the more stringent the constraint on the average transmission probability of each EHS,
as expected. Moreover, the larger $P(\eta^*)$ or $G(\eta^*)$, \emph{i.e.}, the larger the transmission probability or the reward of the other EHSs,
 the larger the Lagrange multiplier $\Lambda(\eta^*)$, thus imposing a stricter constraint 
on EHS $u$ to limit its transmissions and reduce collisions.
Note that the optimal Lagrange multiplier $\Lambda(\eta^*)$ is a function of  the SNE (\ref{NE}),
hence its value optimally balances the power constraints for the EHSs.

In order to solve (\ref{NE}), we address the more general optimization problem, for $\lambda\geq 0$,
\begin{align}\label{etalambda}
\eta^{(\lambda)}=\arg\max_{\eta\in\mathcal U}Z_\lambda(\eta),
\end{align}
where we have defined $Z_\lambda(\eta)\triangleq G(\eta)-\lambda P(\eta)=\sum_{e=1}^{e_{\max}}\pi_\eta(e)z_\lambda(\eta(e))$
and $z_\lambda(x)=g(x)-\lambda x$.
Theorem \ref{propeta} presents structural properties of $\eta^{(\lambda)}$, which are thus inherited by the SNE (\ref{NE}).
\begin{thm}\label{propeta}$ $\\
P1) $\eta^{(\lambda)}$ is unique, \emph{i.e.},
 $Z_\lambda(\eta^{(\lambda)})>Z_\lambda(\eta),\forall\eta\neq\eta^{(\lambda)}$;

\noindent P2) $\eta^{(\lambda)}$ is continuous in $\lambda$;

\noindent P3) $\eta^{(\lambda)}(e)$ is a strictly increasing function of $e$, and
$\eta^{(\lambda)}(e)\in (\eta_{L}(\lambda),\eta_H(\lambda))$, $\forall e\neq 0$, where $\eta_L(\lambda)\in (0,\min\{x_\lambda^*,\beta(s)\})$, $\eta_U(\lambda)\in [\min\{\beta(s),x_\lambda^*\},x_\lambda^*]$\footnote{Note that, if $\beta(s)\geq x_\lambda^*$, then $\eta_U(\lambda)=x_\lambda^*$.} uniquely solve
\begin{align*}
&g(\eta_L(\lambda)){+}(1{-}\eta_L(\lambda))\bar V(y_{\mathrm{th}}(\eta_L(\lambda)))
{-}\frac{z_\lambda(\min\{x_\lambda^*,\beta(s)\})}{\beta(s)}{=}\lambda, \\
&g(\eta_U(\lambda))-\eta_U(\lambda)\bar V(y_{\mathrm{th}}(\eta_U(\lambda)))-z_\lambda(\min\{x_\lambda^*,\beta(s)\})=0,
\end{align*}
 and $x_\lambda^*=\arg\max_{x\in[0,1]} z_\lambda(x)$ is the unique solution in $[0,1]$ of $\bar V(y_{\mathrm{th}}(x))=\lambda$.\footnote{Note
  that $x_\lambda^*$, as defined here, is not to be confused with $x^*$, defined in Prop. \ref{upbound}.}
\end{thm}
\begin{proof}
See Appendix~B.
\end{proof}
The policy $\eta^{(\lambda)}$ can be determined numerically 
using the following PIA, by exploiting the structural properties of Theorem~\ref{propeta}.\footnote{In
 order to avoid confusion with the superscript $(\lambda)$ for the Lagrange multiplier in (\ref{etalambda}), we use the superscript $[i]$ to denote the iteration index of the algorithm.}
\begin{algo}[PIA for a given scenario state $s\in\mathcal S$]\label{PIA}
$ $ \\
1) {\bf Initialization}: $\eta^{[0]}\in\mathcal U$ such that  $\eta^{[0]}(e)\in(\eta_L(\lambda),\eta_H(\lambda))$,
$\eta^{[0]}(e{+}1){>}\eta^{[0]}(e)$, $\forall e\neq 0$; $\epsilon_{PIA}\ll 1$; $i=0$;

\noindent 2) {\bf Policy Evaluation}: compute $Z_\lambda(\eta^{[i]})$ as  in (\ref{GP}), where the 
steady state distribution $\pi_{\eta^{[i]}}(e)$ is given by Prop. \ref{SSD}. Let $D_{\eta^{[i]}}(0)=0$ and 
compute recursively, for $e>0$,
\begin{align}
\label{RVF}
D_{\eta^{[i]}}(e)=&\frac{Z_\lambda(\eta^{[i]})-z_\lambda(\eta^{[i]}(e-1))}{\beta(s)(1-\eta^{[i]}(e-1))}
\nonumber\\&
+\frac{(1-\beta(s))\eta^{[i]}(e-1)}{\beta(s)(1-\eta^{[i]}(e-1))}D_{\eta^{[i]}}(e-1);
\end{align}

\noindent 3) {\bf Policy Improvement}: determine a new policy as
 \begin{align}\label{maxPI}
\eta^{[i+1]}(e)=
\left\{
\begin{array}{ll}
\eta_H(\lambda), & u(e,\lambda)>\eta_H(\lambda),\\
u(e,\lambda), &u(e,\lambda)\in[\eta_L(\lambda),\eta_H(\lambda)],\\
\eta_L(\lambda), & u(e,\lambda)<\eta_H(\lambda),
\end{array}
\right.
\end{align}
where $u(e,\lambda)$ uniquely solves $\bar V(y_{\mr{th}}(u(e,\lambda)))=x^{[i]}(e,\lambda)$,
where 
 \begin{align}
 \label{rho}
& x^{[i]}(e,\lambda){=}\lambda{+}\beta(s) D_{\eta^{[i]}}(e{+}1){+}(1{-}\beta(s))D_{\eta^{[i]}}(e),e{<}e_{\max},
\nonumber
\\
&x^{[i]}(e_{\max},\lambda)=\lambda+(1-\beta(s))D_{\eta^{[i]}}(e_{\max});
 \end{align}

\noindent 4) {\bf Termination test}:
 If $|Z_\lambda(\eta^{[i+1]})-Z_\lambda(\eta^{[i]})|<\epsilon_{PIA}$, return the optimal policy $\eta^{(\lambda)}=\eta^{[i+1]}$;
 otherwise, update the counter $i:=i+1$ and repeat from step 2).
\end{algo}
\begin{remark}
\label{complepia}
Algorithm \ref{PIA} can be computed efficiently.
In fact, 
the steady state distribution $\pi_{\eta^{[i]}}(e)$
in the policy evaluation step 2) 
can be computed in closed form from Prop. \ref{SSD}, whose complexity scales as $O(e_{\max})$.
$D_{\eta^{[i]}}(e)$ is then computed recursively, so
that the resulting complexity of step 2) scales as  $O(e_{\max})$.
The policy improvement step 3) requires to compute, for each energy level $e$,
$u(e,\lambda)$ as the unique solution of $\bar V(y_{\mr{th}}(u(e,\lambda)))=x^{[i]}(e,\lambda)$.
This can be done efficiently using the bisection method \cite{bisection}. If the desired 
accuracy in the evaluation of $u(e,\lambda)$ is $\epsilon_u$, then approximately $-\log_2(\epsilon_u)$ bisection steps are sufficient,
so that the overall complexity of step 2) scales as $O(-e_{\max}\log_2(\epsilon_u))$.
Steps 2)-3) are then iterated $N_{PIA}(\epsilon_{PIA})$ times, depending on the desired accuracy $\epsilon_{PIA}$.
Typically, $N_{PIA}(\epsilon_{PIA})\sim 5\div 10$, so that the overall complexity of Algorithm \ref{PIA}
 scales as $O(N_{PIA}(\epsilon_{PIA})e_{\max}(1-\log_2(\epsilon_u)))$.
 For a further discussion on the complexity and convergence of the PIA, please refer to \cite{Bertsekas2005} and references therein.
\end{remark}

\noindent The following proposition states the optimality of Algorithm~\ref{PIA}.
\begin{propos}\label{PIAoptimality}
Algorithm \ref{PIA} determines the optimal policy  $\eta^{(\lambda)}$ of (\ref{etalambda}) when $\epsilon_{PIA}\to 0$.
\end{propos}
\begin{proof}
The optimality of the policy iteration algorithm is proved in \cite{Bertsekas2005}.
The specific forms of the
 policy evaluation and improvement steps are proved in Appendix~C.
\end{proof}

By comparing (\ref{NE}) and (\ref{etalambda}), $\eta^*$ is optimal for (\ref{NE})
if and only if there exists some $\lambda^*\geq 0$ such that $\Lambda(\eta^{(\lambda^*)})=\lambda^*$.
If such $\lambda^*$ exists, then $\eta^*=\eta^{(\lambda^*)}$.

Note that the SNE may not exist in general. 
This is the case, for instance, if for every symmetric policy $\bs\eta$,
 with $\eta_u=\eta_i,\ \forall u,i$, there exists some unilateral deviation from it
 which improves the network utility, so that the optimality condition (\ref{NEor}) may not hold under any symmetric
$\bs\eta$.
Alternatively, the SNE may not be unique in general, \emph{i.e.}, there may exist multiple symmetric policy profiles satisfying the condition  (\ref{NEor}).
Even though the existence and uniqueness of the SNE, solution of (\ref{NE}), are not granted in general, 
these properties hold for the model considered in this paper, as stated by the following theorem.
\begin{thm}\label{thm1}
There exists a unique $\eta^*\in\mathcal U$ solution of (\ref{NE}), \emph{i.e.},
$\exists!\ \eta^*\in\mathcal U$ such that 
\begin{align*}
G(\eta^*)-\Lambda(\eta^*)P(\eta^*)
>
G(\eta)-\Lambda(\eta^*)P(\eta),\ \forall\eta\neq \eta^*.
\end{align*}
Moreover, under the optimal SNE $\eta^*$, $P(\eta^*)\leq\min\{\beta(s),\frac{1}{U}\}$.
\end{thm}
\begin{proof}
See Appendix~D.
\end{proof}
Note that $P(\eta^*)\leq\frac{1}{U}$. This is due to 
the concurrent transmissions of the $U$ EHSs in the network. In fact,
the (average long-term) probability of successful data delivery to the FC by EHS $u$ is given by
 $P(\eta)[1-(1-P(\eta))^{U-1}]$. This is a decreasing function of $P(\eta)$, for
 $P(\eta)>\frac{1}{U}$, and therefore
 any policy $\eta$ such that $P(\eta)>\frac{1}{U}$ is sub-optimal, 
 resulting in a larger energy expenditure  (hence, less energy availability at the EHSs) and fewer successful transmissions to the FC.

In general, the SNE may yield a sub-optimal solution with respect to the original optimization problem
over symmetric policies
(\ref{origopt}). In fact, although the SNE imposes symmetric policies, it
is optimal only with respect to \emph{unilateral} deviations from it (which do not preserve symmetry), not symmetric neighborhoods.
In contrast, if the policy employed by each EHS is
modified by the same quantity,
thus enforcing symmetry,
 then the resulting network utility may improve, so that the SNE may not be the globally optimal symmetric policy.
Therefore, it may occur that a symmetric policy is the SNE but is not a local/global optimum of (\ref{origopt}).
Similarly, the optimal solution of (\ref{origopt}), $\eta^*$, may not be a SNE.
In fact, $\eta^*$ is such that any deviation
from it over the space of \emph{symmetric policies} yields a smaller network utility.
However, for such symmetric policy $\eta^*$, there may exist some
\emph{unilateral} deviations, occurring over the larger space of non-symmetric policies $\bs\eta\in\mathcal U^U$,
yielding an improved network utility, so that $\eta^*$ may not be the SNE in this case.
Even though the SNE, in general, may not be globally/locally optimal with respect to (\ref{origopt}),
in the following theorem we show that, in fact, it is a \emph{local} optimum of (\ref{origopt}), \emph{i.e.}, any \emph{symmetric} deviation in the neighborhood of the SNE yields 
a smaller network utility.
In the special case where the objective function in (\ref{origopt}) is a concave (or quasi-concave \cite{Boyd})  function of $\eta$,
it follows that the SNE is a global optimum for the original optimization problem (\ref{origopt}).
\begin{thm}\label{thm2}
The SNE $\eta^*$ in (\ref{NE}) is a \emph{local} optimum for (\ref{origopt}).
\end{thm}
\begin{proof}
See Appendix~E.
\end{proof}

To conclude, we present an algorithm to determine the SNE $\eta^*$ in
(\ref{NE}), hence a local optimum of (\ref{origopt}), as proved in Theorem \ref{thm2}.
In particular, we employ the bisection method \cite{bisection}
to determine the unique $\lambda^*$ such that $h(\lambda^*)=0$, where we have defined $h(\lambda)=\Lambda(\eta^{(\lambda)})-\lambda$.
Such algorithm uses Algorithm \ref{PIA} to compute $\eta^{(\lambda)}$.
The SNE $\eta^*$ is then determined as $\eta^*=\eta^{(\lambda^*)}$.
 
Given lower and upper bounds $\lambda_{\min}$ and $\lambda_{\max}$ of $\lambda^*$,
 the bisection method consists in iteratively evaluating the sign of $h(\lambda)$ for $\lambda=(\lambda_{\min}+\lambda_{\max})/2$,
 based on which $\lambda_{\min}$ and $\lambda_{\max}$ of $\lambda^*$ can be refined,
until the desired accuracy is attained.
Its optimality is a consequence of Prop. \ref{Lembdadecreasing} in Appendix~D:
 if $h(\lambda)>0$ (respectively, $h(\lambda)<0$) for some $\lambda$, then necessarily $\lambda<\lambda^*$ ($\lambda>\lambda^*$).
The lower bound is initialized as $\lambda_{\min}=0$.
Since $P(\eta^*)\leq \min\{\beta(s),\frac{1}{U}\}$  (Theorem \ref{thm1}), hence $G(\eta^*)<g(P(\eta^*))\leq g(\min\{\beta(s),\frac{1}{U}\})$,
from (\ref{Lambda}) $\lambda_{\max}$ is initialized as
\begin{align}\label{lambdamax}
\!\!\!\lambda^*{=}\Lambda(\eta^{(\lambda^*)}){<}
\min\left\{\!\frac{U{-}1}{1{-}\beta(s)}g(\beta(s)),Ug\left(\!{\frac{1}{U}}\!\right)\!\right\}{=}\lambda_{\max}.\!\!
\end{align}
The bisection algorithm, coupled with Algorithm \ref{PIA} to determine $\eta^{(\lambda)}$, is given as follows.
\begin{algo}[Optimal $\lambda^*$ via bisection algorithm]\label{algoSNE}
$ $ \\
1) {\bf Initialization}: accuracy of bisection method $\epsilon_{BM}>0$ and of the PIA $\epsilon_{PIA}>0$, $\lambda_{\min}^{[0]}=0$ and $\lambda_{\max}^{[0]}$ as in (\ref{lambdamax}); $i=0$;

\noindent 2)
 {\bf Policy optimization}: Let $\lambda^{[i]}=(\lambda_{\min}^{[i]}+\lambda_{\max}^{[i]})/2$;
determine $\eta^{(\lambda^{[i]})}$ using the PIA;

\noindent 3) {\bf Evaluation}: Evaluate  $h(\lambda^{[i]})=~\Lambda(\eta^{(\lambda^{[i]})})-\lambda^{[i]}$ under the policy $\eta^{(\lambda^{[i]})}$;

\noindent 4) {\bf Update of lower \& upper bounds}:
\begin{itemize}
 \item  If $h(\lambda^{[i]})\geq 0$, 
update the lower and upper  bounds as
\begin{align}
\left\{
\begin{array}{l}
\lambda_{\min}^{[i+1]}=\lambda^{[i]},
\\
\lambda_{\max}^{[i+1]}=\min\{\lambda_{\max}^{[i]},\Lambda(\eta^{(\lambda^{[i]})})\};
\end{array}\right.
\end{align}
 \item If $h(\lambda^{[i]})<0$, 
update the  lower and upper bounds as
\begin{align}
\left\{\begin{array}{l}
\lambda_{\min}^{[i+1]}=\max\{\lambda_{\min}^{[i]},\Lambda(\eta^{(\lambda^{[i]})})\},
\\
\lambda_{\max}^{[i+1]}=\lambda^{[i]};
\end{array}\right.
\end{align}
\end{itemize}

\noindent 5)
 {\bf Termination test}:
If $\lambda_{\max}^{[i+1]}-\lambda_{\min}^{[i+1]}<\epsilon_{BM}$, return the optimal policy $\eta^*=\eta^{(\lambda^{[i]})}$;
 otherwise, update the counter $i:=i+1$ and repeat from step 2).
\end{algo}
\begin{remark}
Note that step 4) updates both $\lambda_{\min}$ and $\lambda_{\max}$.
This is because $h(\lambda)$ is a decreasing function of
$\lambda$ and $\Lambda(\eta^{(\lambda)})$ is a non-increasing function of $\lambda$ (Prop.~\ref{Lembdadecreasing} in Appendix~D). It follows that,
 if $h(\lambda){>}0$, then $\lambda{<}\lambda^*{=}\Lambda(\eta^{(\lambda^*)}){\leq}\Lambda(\eta^{(\lambda)})$,
 so that the upper bound can be updated as $\lambda_{\max}^{[i+1]}=\min\{\lambda_{\max}^{[i]},\Lambda(\eta^{(\lambda^{[i]})})\}$.
Similarly, if  $h(\lambda){<}0$, then $\lambda{>}\lambda^*{=}\Lambda(\eta^{(\lambda^*)}){\geq}\Lambda(\eta^{(\lambda)})$,
 so that the lower bound can be updated as $\lambda_{\min}^{[i+1]}=\max\{\lambda_{\min}^{[i]},\Lambda(\eta^{(\lambda^{[i]})})\}$.
\end{remark}

\begin{remark}
\label{compleSNE}
The maximum  number of iterations of Algorithm  \ref{algoSNE} in order to achieve
the desired accuracy $\epsilon_{BM}$ is
$\log_2(\lambda_{\max}^{[0]}/\epsilon_{BM})$. By combining it with Algorithm \ref{PIA}, the overall complexity
to compute the SNE
 thus scales as
$O(\log_2(\lambda_{\max}^{[0]}/\epsilon_{BM})N_{PIA}(\epsilon_{PIA})e_{\max}(1-\log_2(\epsilon_u)))$.
\end{remark}
\section{Heuristic policy}\label{heur}
In this section, we evaluate the following heuristic policy, 
inspired by the upper bound (\ref{ub2}), 
\begin{align}
\label{heurscheme}
\eta^{(H)}(e,s)&=
\underset{x\in[0,\beta(s)]}{\arg\max} Ug(x)(1-x)^{U-1}
\nonumber\\&
=
\min\{x^*,\beta(s)\},\ \forall e>0,
\end{align}
where $x^*\in(0,1/U)$ is defined in Prop. \ref{upbound}.
From (\ref{GP}), we obtain
\begin{align}
&G(\eta^{(H)}|s)=(1-\pi_{\eta^{(H)}}(0|s))g(\min\{x^*,\beta(s)\}),\\ & 
 P(\eta^{(H)}|s)=(1-\pi_{\eta^{(H)}}(0|s))\min\{x^*,\beta(s)\},
\end{align}
where $\pi_{\eta^{(H)}}(0)$ can be derived using Prop. \ref{SSD} as
\begin{align*}
&\pi_{\eta^{(H)}}(0|s)=
\left\{\begin{array}{ll}
\frac{1}{
1+\beta(s)\frac{\left[\frac{\beta(s)(1-x^*)}{(1-\beta(s))x^*}\right]^{e_{\max}}-1}{\beta(s)-x^*}
}, & x^*<\beta(s),\\
\frac{1-\beta(s)}{e_{\max}+1-\beta(s)}, & x^*\geq\beta(s).
\end{array}\right.
\end{align*}
Note that, by letting $e_{\max}\to\infty$, we obtain $\pi_{\eta^{(H)}}(0|s)\to 0$, hence 
$G(\eta^{(H)}|s)\to g(\min\{x^*,\beta(s)\})$ and
$P(\eta^{(H)}|s)\to\min\{x^*,\beta(s)\}$.
Therefore, in the limit of asymptotically large battery capacity
we obtain $\lim_{e_{\max}\to\infty}R_{\bs\eta^{(H)}}=R^{(up)}$,
\emph{i.e.}, the upper bound $R^{(up)}$ given by Prop. \ref{upbound} is asymptotically achieved by the heuristic policy.
It follows that the heuristic policy is asymptotically optimal for large battery capacity.
Note that the heuristic policy is similar to the policies developed in \cite{ZlatanovSH13} for general EH networks and
infinite battery capacity under the assumption of a centralized controller. Therein, it has been shown that
the asymptotically optimal power allocation for the EH network is identical to the optimal power allocation for an equivalent non-EH network whose nodes have infinite energy available and employ the same average transmit power as the corresponding nodes in the EH network.
Similarly, the heuristic policy (\ref{heurscheme}) maximizes the upper bound on the network utility in scenario $s$, $Ug(x)(1-x)^{U-1}$, achievable for asymptotically large battery capacity $e_{\max}\to\infty$ under the policy
$\eta(e,s)=x,\ \forall e$,
and replacing the EH operation of the EHSs with an average
energy expenditure constraint equal to the average EH rate, so that $x\in[0,\beta(s)]$ in (\ref{heurscheme}).

We identify the following operational scenarios.
\subsection{{Energy-limited scenario} $x^*\geq \beta(s)$ (hence $1/U\geq \beta(s)$)}
\label{enelim}
In this case, the expected amount of energy harvested by the EHSs is small compared to the optimal
energy expenditure $x^*$, which would optimize the overall utility collected at the FC $Ug(x^*)(1-x^*)^{U-1}$ (under the ideal case of
infinite energy availability); it follows that energy is scarce and the shared multiaccess channel is under-utilized.
If this holds for all scenarios $s\in\mathcal S$,
we use the bounds
\begin{align}
&\pi_{\eta^{(H)}}(0|s)=\frac{1-\beta(s)}{e_{\max}+1-\beta(s)}<\frac{1}{e_{\max}}, 
\\&
P(\eta^{(H)}|s)<\min\{x^*,\beta(s)\}=\beta(s),
\end{align}
to obtain
\begin{align}
\left(1-\frac{1}{e_{\max}}\right)
R^{(up)}<R_{\bs\eta^{(H)}}<R^{(up)},
\end{align}
so that the upper bound is approached with rate inversely proportional to the battery capacity.
\subsection{{Network-limited scenario} $x^*<\beta(s)$}
\label{netlim}
In this case, the expected amount of energy received by the EHSs is large compared 
 to the optimal energy expenditure $x^*$, hence energy is abundant and the shared
 multiaccess channel is the limiting factor, resulting in under-utilization of the harvested energy.
 If this holds for all scenarios $s\in\mathcal S$,
 letting $\alpha(s)\triangleq \ln\left(\frac{\beta(s)(1-x^*)}{(1-\beta(s))x^*}\right)>0$, 
 we use the bounds
 \begin{align}
& \pi_{\eta^{(H)}}(0|s)
{=}
 \frac{1}{
1{+}\beta(s)\frac{\left[\frac{\beta(s)(1-x^*)}{(1-\beta(s))x^*}\right]^{e_{\max}}-1}{\beta(s)-x^*}
}
{=}
 \frac{1}{1{+}\beta(s)\frac{e^{\alpha(s)e_{\max}}-1}{\beta(s)-x^*}}
\nonumber\\&
=
 \frac{1-x^*/\beta(s)}{e^{\alpha(s)e_{\max}}-x^*/\beta(s)}
<
\max_{t\in[0,1]} \frac{1-t}{e^{\alpha(s)e_{\max}}-t}
=
e^{-\alpha(s)e_{\max}},
\nonumber
\\&
P(\eta^{(H)}|s)<\min\{x^*,\beta(s)\}=x^*,
 \end{align}
to obtain
\begin{align}
\left(1-\sum_s\pi_S(s)e^{-\alpha(s) e_{\max}}\right)R^{(up)}<
R_{\bs\eta^{(H)}}<R^{(up)},
\end{align}
so that the upper bound is approached exponentially fast.

In the more general case where
$x^*\geq \beta(s)$ for $s\in\mathcal S_{EL}\subseteq\mathcal S$ and
$x^*<\beta(s)$ for $s\in\mathcal S_{NL}\equiv\mathcal S\setminus\mathcal S_{EL}$,
where $\mathcal S_{EL}$ and $\mathcal S_{NL}$ are the sets of scenario states
corresponding to energy-limited or network-limited operation, respectively, we thus obtain
\begin{align}
\label{guarantee}
&\left(1-\sum_{s\in\mathcal S_{EL}}\pi_S(s)\frac{1}{e_{\max}}-\sum_{s\in\mathcal S_{NL}}\pi_S(s)e^{-\alpha(s) e_{\max}}\right)R^{(up)}
\nonumber\\&
<
R_{\bs\eta^{(H)}}<R^{(up)},
\end{align}
so that a trade-off between inversely decaying and exponentially decaying gap to the upper bound is achieved.

\begin{remark}
\label{prqactiact}
Due to the resource constraints of EHSs, 
an efficient implementation of the proposed scheme
can be achieved by solving
Algorithm \ref{algoSNE} at the FC in a centralized fashion.
The SNE $\eta(\cdot,s)$ for a given scenario $s$ returned by Algorithm \ref{algoSNE} can then
be broadcast
to the EHSs, which in turn store in a look-up table the corresponding values of the threshold $y_{\mr{th}}(e,s)$
as a function of the energy level $e$ and common scenario $s$.
The complexity of Algorithm \ref{algoSNE} is discussed in Remarks \ref{complepia} and \ref{compleSNE}.
Alternatively, a lower complexity is obtained by employing the heuristic scheme (\ref{heurscheme}).
In this case, the EHSs need to store only a threshold value $y_{\mr{th}}^{(H)}(s)$, independent of the current energy level $e$.
As shown 
analytically by (\ref{guarantee}) and numerically in Sec. \ref{Sec:numer}, such scheme achieves near-optimal performance
for sufficiently large battery capacity ($e_{\max}\simeq 10$).
Algorithm~\ref{algoSNE} or the heuristic scheme needs to be computed sporadically at the FC, whenever the environmental or network conditions change
(utility distribution $f_V(\cdot)$ and $f_{Y|V}(\cdot |v)$, average EH rate $\beta(s)$, battery storage capacity $e_{\max}$, \emph{e.g.}, due to battery degradation effects \cite{nicolo-exinfocom}, or network size $U$).
\end{remark}

\section{Numerical Results}\label{Sec:numer}
In this section, we present some numerical results for different settings, by varying 
 the battery capacity $e_{\max}\in\{1,10\}$ and the EH rate $\beta\in\{1/U,0.1,0.01\}$.
 We consider a static scenario setting $S_k=s,\forall k$, and set $\rho(y)=0,\ \forall y$, so that
packet losses are due only to collisions, but not to fading.
Note that in this case the choice of the fading model only affects the form of the expected utility $\bar V(y)$ in (\ref{condexp}).
We model $V_{u,k}$ as an exponential random variable with unit mean,
 with pdf $f_{V}(v)=e^{-v},v{\geq}0$, and assume that it is observed perfectly, 
 \emph{i.e.}, $Y_{u,k}=V_{u,k}$. 
From (\ref{eta}) and (\ref{gx}), we obtain $g(x)=x(1-\ln x)$.
 We remark that other choices of the distribution of $V_{u,k}$ and of $Y_{u,k}$
 only affect the shape of the function $g(x)$, but do not yield any additional
 insights with respect to the specific case considered. 
We evaluate the performance of the following policies: the SNE derived via Algorithm \ref{algoSNE};
the heuristic policy (HEUR) analyzed in Sec.~\ref{heur};
the \emph{energy-balanced policy} (EBP), where each EHS, in each slot, transmits with probability $\beta$ equal to the average EH rate;
and the \emph{network-balanced policy} (NBP), where each EHS, in each slot, transmits with probability $1/U$ so as to maximize the expected throughput.
For $e_{\max}=1$, we plot also the \emph{globally optimal policy} (GOP) solution of (\ref{origopt}). 
In fact, in this case, the policy is defined by $\eta(1)\in [0,1]$, and can be easily optimized via an exhaustive search,
thus serving as an upper bound for the case $e_{\max}=1$.
On the other hand, for $e_{\max}=10$, we plot the upper bound (UB) given by Prop. \ref{upbound}.
Note that, in general, both SNE and HEUR are sub-optimal with respect to (\ref{origopt}), hence GOP and UB provide 
upper bounds for the cases $e_{\max}=1$ and $e_{\max}=10$, respectively, against which the performance of SNE can be evaluated.

\begin{figure}[t]
\centering
\includegraphics[width =\linewidth,trim = 18mm 4mm 2mm 10mm,clip=true]{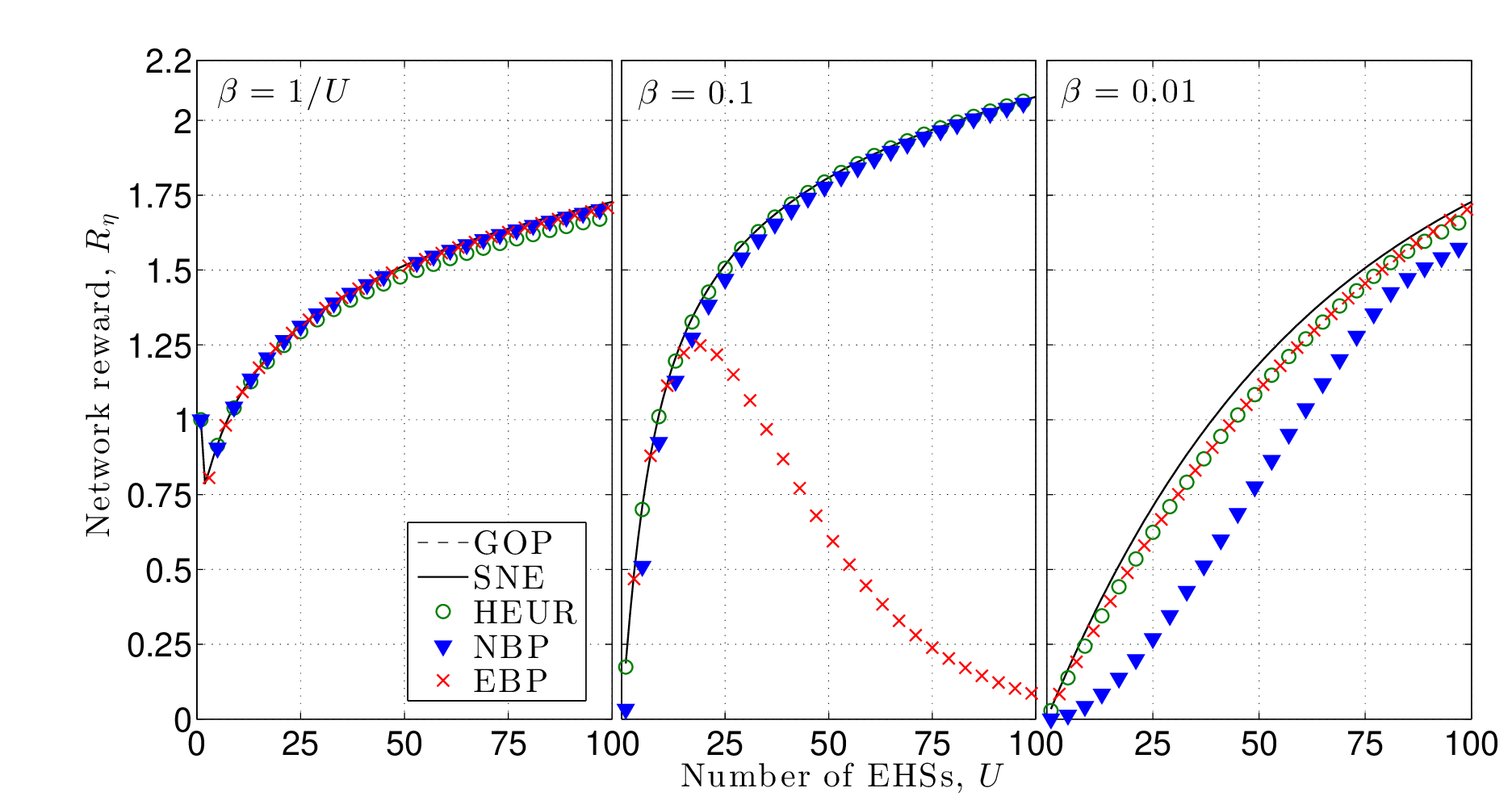}
\caption{Network utility (\ref{eq:averageGainTH}) for different EH rates $\beta\in\{1/U,0.1,0.01\}$,
as a function of the network size $U$. $e_{\max}=1$.}
\label{fig:R2}
\end{figure}

In Figs. \ref{fig:R2} and \ref{fig:R3}, we plot the network utility (\ref{eq:averageGainTH}),
as a function of the number of EHSs $U$, for the cases $e_{\max}=1$ and $e_{\max}=10$, respectively.
 We note that, when $e_{\max}=1$, SNE attains the upper bound given by GOP, hence it is optimal for this case.
 On the other hand, the performance degradation of SNE with respect to UB for the case 
 $e_{\max}=10$ is within 3\%, hence SNE is near-optimal for this case (note that UB is approached only asymptotically for $e_{\max}\to\infty$).
 This suggests that the local optimum achieved by the SNE closely approaches the globally optimal policy.
 A performance degradation is incurred by HEUR with respect to SNE. In particular, HEUR performs within
 18\% of GOP for $e_{\max}=1$, and within 9\% of UB  for $e_{\max}=10$, hence it provides a good heuristic scheme, which achieves both satisfactory performance
and low complexity. Note also that HEUR does not require knowledge of the battery energy level, hence it is suitable for those scenarios where the state of the battery is
unknown or
imperfectly observed \cite{MichelusiFusion}.
By comparing EBP and NBP, we notice that the former performs better than the latter in the energy-limited scenario $\beta\leq 1/U$,
and worse in the network-limited scenario $\beta>1/U$.\footnote{Note that the scenario $x^*<\beta\leq1/U$ is indeed \emph{network-limited} (see Sec. \ref{enelim}).
However, both EBP and NBP, unlike HEUR, transmit with probability larger than $x^*$ in this case, so that their operation is energy-limited (see Sec. \ref{netlim}).
This case is thus included in the \emph{energy-limited scenario} for convenience.}
In fact, EBP leads to over-utilization of the channel in the network-limited scenario $\beta>1/U$, hence to frequent collisions and poor performance.
On the other hand, NBP matches the transmissions to the channel, without taking into account the energy constraint induced by the EH mechanism,
hence it performs poorly in the energy-limited scenario $\beta\leq 1/U$, where energy is scarce.

Interestingly, in all cases considered,
 the network utility under the SNE increases with the number of EHSs $U$. 
This behavior is due to the strict concavity of $g(x)$, such that
a diminishing return is associated to a larger transmission probability $x$.
Therefore, the smaller the number of EHSs $U$, the more the transmission opportunities
for each EHS, but the smaller the marginal gain, so that the network utility decreases. 
 \begin{figure}[t]
\centering
\includegraphics[width =\linewidth,trim = 18mm 4mm 2mm 10mm,clip=true]{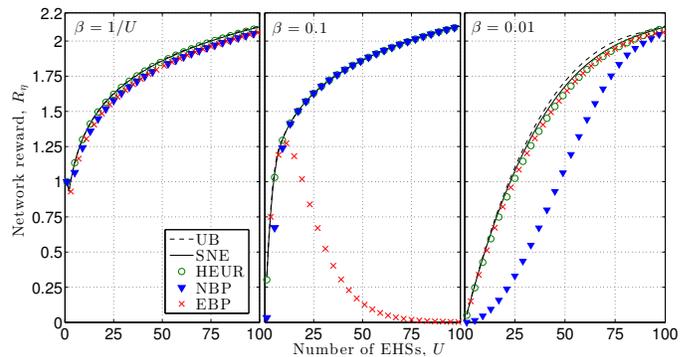}
\caption{Network utility (\ref{eq:averageGainTH}) for different EH rates $\beta\in\{1/U,0.1,0.01\}$,
as a function of the network size $U$. $e_{\max}=10$.}
\label{fig:R3}
\end{figure}
For $U<10$, SNE achieves the best performance when $\beta=1/U$,
since more energy is available to the EHSs than for $\beta\in\{0.1,0.01\}$.
In contrast, for $U>10$ and $e_{\max}=10$, SNE performs best in both cases
 $\beta=1/U$ and $\beta=0.1$ corresponding to the network-limited scenario, despite a larger energy availability in the latter case.
This is due to the fact that, as proved in Theorem \ref{thm1},
under the SNE, $P(\eta^*)\leq\min\{\beta,1/U\}=1/U$,
hence the performance bottleneck is due to the number of EHSs in the network, rather than
to the energy availability. In the case $\beta=0.1$,
on average, $\beta-P(\eta^*)\geq\beta-1/U$ energy quanta are lost in each slot due to overflow, in order to limit the collisions.
In contrast, when $\beta=0.01$ (energy-limited scenario), we have $P(\eta^*)\leq\min\{\beta,1/U\}=\beta$
for all values of $U$ considered, hence the performance bottleneck is energy availability.
A different trend is observed when $e_{\max}=1$. In this case, for $U>10$,
SNE performs significantly better in
scenario $\beta=0.1$  than in scenario $\beta=1/U$. This
follows from the fact that, when $e_{\max}=1$, whenever an EHS transmits,
its battery becomes empty, hence it enters a recharge phase, with expected duration $1/\beta=U$,
during which it is inactive.
In contrast, the recharge phase in scenario $\beta=0.1$ is faster and equals $10<U$ slots, 
so that each EHS becomes more quickly available for data transmission.
This phenomenon is negligible when $e_{\max}=10$, since the larger battery capacity 
minimizes the risk of energy depletion, resulting in the impossibility to transmit for the EHSs.

\begin{figure}[t]
\centering
\includegraphics[width = \linewidth,trim = 0mm 0mm 0mm 5mm,clip=true]{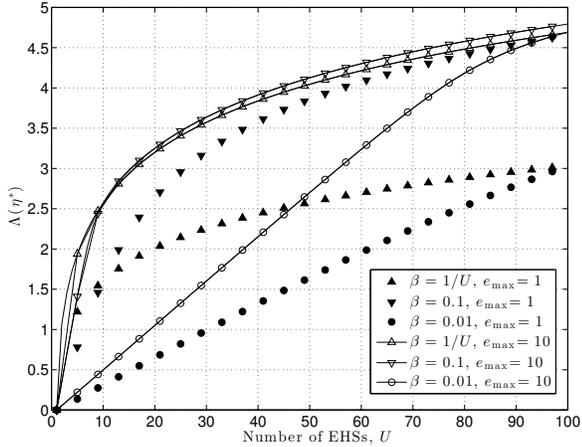}
\caption{$\Lambda(\eta^*)$ under the SNE $\eta^*$ for different 
values of the battery capacity
 $e_{\max}\in\{1,10\}$
and for different EH rates $\beta\in\{1/U,0.1,0.01\}$,
as a function of the network size $U$.}
\label{fig:convergence}
\end{figure}

In Fig. \ref{fig:convergence}, we plot the optimal Lagrange multiplier $\lambda^*=\Lambda(\eta^*)$ versus the number of EHSs $U$.
We notice that the larger the number of EHSs, the larger $\lambda^*$. In fact, for a given policy $\eta$,
the larger $U$, the more frequent the collisions. A larger $\lambda^*$ thus balances
this phenomenon by penalizing the average transmission probability $P(\eta^*)$  in (\ref{NE}),
and in turn forces each EHS to transmit more sparingly, so as to accommodate the transmissions
of more nodes in the network.
Moreover, the larger the EH rate $\beta$, the larger $\lambda^*$. In fact, the larger $\beta$, the
larger the energy availability for each EHS, which could, in principle, transmit more frequently and, at the same time, cause
more collisions. The effect of a larger $\beta$ (having more transmissions, hence more collisions in the system)
is thus balanced by a larger $\lambda^*$, which penalizes transmissions.

\section{Conclusions}\label{Sec:concl}
In this paper, we have considered a WSN of EHSs which randomly access
a collision channel, to transmit packets of random utility to a common FC,
based on local information about the current energy availability and an estimate of the utility of the current data packet.
We have studied the problem of designing decentralized random access policies so as to maximize the overall network utility, defined as the
average long-term aggregate network utility of the data packets
successfully reported to the FC.
We have formulated the optimization problem as a game, where each EHS unilaterally maximizes
the network utility, and characterized the symmetric Nash equilibrium (SNE).
We have proved the existence and uniqueness of the SNE, showing that it is a local optimum
of the original optimization problem, and we have derived an algorithm to compute it, and 
a heuristic policy which is asymptotically optimal for large battery capacity.
 Finally, we have presented some numerical results, showing that the proposed SNE typically achieves near-optimal performance with respect to the globally optimal policy,
 at a fraction of the complexity, and we have identified two operational regimes, the \emph{energy-limited} and \emph{network-limited} scenarios,
  where the bottleneck of the system
is given by the energy and the shared multiaccess scheme, respectively.
\section*{Appendix~A: Proof of Prop.~\ref{lemma:thresStruct}}
Let $\mu_u$ be a stationary randomized policy used by EHS $u$, and $\mathcal R_{\mu_u}$ be the set 
of stationary randomized 
policies which induce, in each energy level $e_u\in\mathcal E$ and scenario $s\in\mathcal S$, the same transmission probability as $\mu_u$, defined with respect to the 
utility observation, \emph{i.e.},  $\forall e_u\in\mathcal E,\forall s\in\mathcal S,\ \forall\tilde\mu_u\in\mathcal R_{\mu_u}$,
\begin{align}
\mathbb E\left[\tilde \mu_u(e_u,Y_u,s)\right]=\mathbb E\left[\mu_u(e_u,Y_u,s)\right].
\end{align}
Then, since $\mu_u\in\mathcal R_{\mu_u}$, from (\ref{eq:averageGain}) the network utility satisfies the inequality
\begin{align}
\label{proofofthstruc}
R_{\mu_u,\bs\mu_{-u}}(\mathbf e_0,s_0)
\leq  \max_{\tilde\mu_u\in\mathcal R_{\mu_u}} R_{\tilde\mu_u,\bs\mu_{-u}}(\mathbf e_0,s_0).
\end{align}

We now show that the maximizer of the right hand side of (\ref{proofofthstruc})
has a threshold structure with respect to the utility observation.
 From (\ref{Gu}) and (\ref{eq:averageGain}), for any $\tilde \mu_u\in\mathcal R_{\mu_u}$, we have
\begin{align*}
&R_{\tilde\mu_u,\bs\mu_{-u}}(\mathbf e_0,s_0)
\nonumber\\&
=
\underset{K \rightarrow \infty}{\lim\inf}\frac{1}{K}
\sum_{k=0}^{K-1}\!
\sum_{\mathbf e\in\mathcal E^U,s\in\mathcal S}\!\!\!\!\!\!\mathbb P(\mathbf E_{k}{=}\mathbf e,S_k{=}s|\mathbf e_0,s_0,\tilde\mu_u,\bs\mu_{-u})
\nonumber\\&\times
\sum_{n}\mbb{E}\left[
\chi(\gamma_{n,k}>\gamma_{\mathrm{th}}(Y_{n,k}))
Q_{n,k}V_{n,k}
\vphantom{\left.\left.\times\prod_{i\neq n}(1-Q_{i,k})\right|\mathbf E_k=\mathbf e,S_k=s,\tilde\mu_u,\bs\mu_{-u}\right].}
\right.
\nonumber\\&
\left.\left.
\times
\prod_{i\neq n}(1-Q_{i,k})\right|\mathbf E_k=\mathbf e,S_k=s,\tilde\mu_u,\bs\mu_{-u}\right].
\end{align*}
Now, using the fact that $Q_{n,k}{=}1$ with probability $\mu_n(E_{n,k},Y_{n,k},S_k)$, independent of
the energy levels and utility observations of the other EHSs, and $(V_{n,k},Y_{n,k},\gamma_{n,k})$ is independent of $\mathbf E_k$ and $S_k$,
and i.i.d. across EHSs,
 we obtain
\begin{align}
\nonumber
&\mbb{E}\!\left[\chi(\gamma_{n,k}>\gamma_{\mathrm{th}}(Y_{n,k}))
Q_{n,k}V_{n,k}
\vphantom{\times\prod_{i\neq n}(1-Q_{i,k})}
\right.
\nonumber\\&
\qquad\left.\left.
\times
\prod_{i\neq n}(1-Q_{i,k})\right|\mathbf E_k=\mathbf e,S_k=s,\tilde\mu_u,\bs\mu_{-u}\right]
\nonumber\\&
=
\mbb{E}\left[\mu(e_n,Y_{n,k},s)\chi(\gamma_{n,k}>\gamma_{\mathrm{th}}(Y_{n,k}))V_{n,k}\right]
\prod_{i\neq n}(1-\eta_i(e_i,s))
\nonumber\\&
=
\mbb{E}\left[\mu(e_n,Y_{n,k})\bar V(Y_{n,k})\right]\prod_{i\neq n}(1-\eta_i(e_i,s)),
\end{align}
where we have defined the transmission probability of EHS $i$ as $\eta_i(e_i,s)\triangleq\mathbb E[\mu_i(e_i,Y_{i,k},s)]$,
and we have used (\ref{condexp}).
Therefore, denoting the steady state probability of $(\mathbf E_k,S_k)=(\mathbf e,s)$ as
$\pi_{\tilde\mu}(\mathbf e,s;\mathbf e_0,s_0)=\underset{K \rightarrow \infty}{\lim\inf}\frac{1}{K}\sum_{k=0}^{K-1}\mathbb P(\mathbf E_{k}=\mathbf e,S_k=s|\mathbf e_0,s_0,\tilde\mu)$,
we obtain
\begin{align}\label{bbb2}
&R_{\tilde\mu_u,\bs\mu_{-u}}(\mathbf e_0,s_0)=
\sum_{\mathbf e\in\mathcal E^U,s\in\mathcal S}
\pi_{\tilde\mu_u,\boldsymbol{\mu}_{-u}}(\mathbf e,s;\mathbf e_0,s_0)
\\&
\nonumber\times
\sum_{n=1}^U\mbb{E}\left[\mu(e_n,Y_{n,k},s)\bar V(Y_{n,k})\right]\prod_{i\neq n}(1-\eta_i(e_i,s)).
\end{align}
Using the fact that $Y_{u,k}$ is independent of $\{\mathbf E_{k},i=0,\dots,k\}$, $\forall k$,
and i.i.d. across EHSs, it can be proved by induction on $k$ that
the probability $\mathbb P(\mathbf E_k=e,S_k=s|\mathbf e_0,s_0,\tilde\mu_u,\bs\mu_{-u})$ depends on $\tilde\mu_u$ only through the expected transmission probability
$\eta_u(e_u,s)$,
which is common to all $\tilde\mu_u\in\mathcal R_{\mu_u}$, and therefore
the steady state distribution of $(\mathbf E_k,S_k)$
is the same for all $\tilde\mu_u\in\mathcal R_{\mu_u}$, \emph{i.e.},
 $\pi_{\tilde\mu_u,\bs\mu_{-u}}(\mathbf e,s;\mathbf e_0,s_0)=\pi_{\mu_u,\bs\mu_{-u}}(\mathbf e,s;\mathbf e_0,s_0)$.
Therefore, from (\ref{proofofthstruc}) and (\ref{bbb2}) we obtain
\begin{align}
\label{above}\nn
&R_{\bs\mu}(\mathbf e_0,s_0)
\leq\max_{\tilde\mu_u\in\mathcal R_{\mu_u}}R_{\tilde\mu_u,\bs\mu_{-u}}(\mathbf e_0,s_0)
\nonumber\\&
=
\sum_{e_u\in\mathcal E,s\in\mathcal S}
\left[\vphantom{\prod_{i\neq n}}\right.
\mbb{E}\left[\mu^*(e_u,Y_{u,k},s)\bar V(Y_{u,k})\right]
\nonumber\\&
\times\sum_{\mathbf e_{-u}\in\mathcal E^{U-1}}\pi_{\mu_u,\boldsymbol{\mu}_{-u}}(\mathbf e,s;\mathbf e_0,s_0)\prod_{i\neq u}(1-\eta_i(e_i,s))
\nonumber
\\&
+(1-\eta_u(e_u,s))
\!\!\!\!\!\!\!\sum_{\mathbf e_{-u}\in\mathcal E^{U-1}}\!\!\!\!\!\!\!\pi_{\mu_u,\boldsymbol{\mu}_{-u}}(\mathbf e,s;\mathbf e_0,s_0)
\nonumber\\&
\left.\quad
\times
\sum_{n\neq u}
\mbb{E}\left[\mu_n(e_n, Y_{n,k},s)V_{n,k}\right]\prod_{i\neq n,u}(1-\eta_i(e_i,s))
\right],
\end{align}
where, for each $e_u{\in}\mathcal E$, $s{\in}\mathcal S$, we have defined
 $\mu_u^*(e_u,\cdot,s)$  as
\begin{align}
\label{conv}
 \mu_u^*(e_u,\cdot,s)=&
 \underset{\tilde\mu(e_u,\cdot,s):\mathbb R^+\mapsto[0,1]}{\arg\max} 
\mbb{E}\left[\tilde\mu_u(e_u, Y_{u,k},s)\bar V(Y_{u,k})\right],
\nonumber\\&
  \mathrm{s.t.\ }\mathbb E[\tilde\mu_u(e_u,Y_{u,k},s)]=\eta_u(e_u,s).
\end{align}
Since (\ref{conv}) is a convex optimization problem, 
we can solve it
using the Lagrangian method \cite{Boyd},
with Lagrangian multiplier $\bar V(y^*_{\mr{th},u}(e_u,s))$ for the constraint $\mathbb E[\tilde\mu_u(e_u,Y_{u,k},s)]=\eta_u(e_u,s)$,  \emph{i.e.},
\begin{align*}
\!\!\mu_u^*\!(e_u{,}\!\cdot\!{,}s)\!\!=\!
\!\!\!\!\!\!\!\!\!\!\!\underset{\tilde\mu(e_u,\cdot,s):\mathbb R^+\mapsto[0,1]}{\arg\!\max}\!\!\!\!\!\!\!\!\!\!
 \mbb{E}\!\left[\tilde\mu_u(e_u{,}Y_{u,k}{,}s)(\bar V(Y_{u,k}){-}\!\bar V(y^*_{\mr{th},u}(e_u{,}s)\!)\!)\!\right]\!\!.
\end{align*}
The solution to this optimization problem is given by
$ \mu_u^*(e_u,y_{u,k},s){=}\chi(\bar V(y_{u,k}){\geq}\bar V(y^*_{\mr{th},u}(e_u,s)))$,
and, since $\bar V(y)$ is a strictly increasing and continuous function of $y$ (see Assumption~\ref{assu}),
we obtain the threshold policy (with respect to the utility observation)
 $\mu_u^*(e_u,y_{u,k},s){=}\chi(y_{u,k}{\geq}y^*_{\mr{th},u}(e_u,s))$,
where $y^*_{\mr{th},u}(e_u,s)$ is such that
$\mathbb E[\mu_u^*(e_u,Y_{u,k},s)]=\mathbb P(Y_{u,k}\geq y^*_{\mr{th},u}(e_u,s))=\eta_u(e_u,s)$.
The threshold structure in (\ref{muthresstruc}) is thus proved.
\hfill$\blacksquare$

\section*{Appendix~B: Proof of Theorem \ref{propeta}}
Due to space constraints, we do not provide an explicit proof of P1 and P2.
Due to the concavity of $g(x)$, $x_\lambda^*=\arg\max_{x\in[0,1]}z_\lambda(x)$
 is the unique solution of $g^\prime(x){=}\lambda$.
Moreover, $G(\eta){<}g(P(\eta))$ from Jensen's inequality \cite{Boyd}, hence
$Z_\lambda(\eta){<}z_\lambda(P(\eta)){\leq}z_\lambda(\min\{x_\lambda^*,\beta\})$, where we have used the fact that the EH operation
enforces the constraint $P(\eta)\leq\beta$.

We first prove that $\eta^{(\lambda)}(e)$ must be a strictly increasing function of the energy level $e$ (P3.A), by using a technique we term
\emph{transmission transfer}.\footnote{As we will see, despite the term \emph{transmission transfer},
 this technique does not necessarily conserve the amount of transmissions transferred among states.
} In the second part, we prove that $\eta^{(\lambda)}(e)\in (\eta_L(\lambda),\eta_H(\lambda)),\ \forall e>0$ (P3.B).
Let $\eta_0\in\mathcal U$ be a generic transmission policy which violates P3.A,
\emph{i.e.}, there exists $\epsilon\in\{1,\dots, e_{\mr{max}}-1\}$ such that
\begin{align}
\eta_0(\epsilon-1)<\eta_0(\epsilon)\geq \eta_0(\epsilon+1).
\end{align}
Note that, since $\eta_0\in\mathcal U$, then $\eta_0(0)=0$, $\eta_0(1)>0$. Then, such $\epsilon$ can be found as the smallest energy level $e\geq 1$
such that $\eta_0(e+1)\leq \eta_0(e)$. If such $\epsilon$ is not found, then necessarily $\eta_0(0)<\eta_0(1)<\dots<\eta_0(e)<\dots<\eta_0(e_{\max})$,
hence P3.A is not violated.
We now define a new transmission policy, $\eta_{\delta}$, parameterized by $\delta>0$, as:
\begin{align}\label{etamod}
\!\!\eta_\delta(e){=}
\left\{
\begin{array}{ll}
 \eta_0(e), & e\in\mathcal E\setminus\{\epsilon-1,\epsilon,\epsilon+1\}\\
 \eta_0(\epsilon-1)+h(\delta), & e=\epsilon-1\\
 \eta_0(\epsilon)-\delta, & e=\epsilon\\
 \eta_0(\epsilon+1)+r(\delta), & e=\epsilon+1,
\end{array}
\right.\!\!
\end{align}
where $h(\delta)>0$ and $r(\delta)>0$ are continuous functions such that $h(0)=r(0)=0$.
Intuitively, policy $\eta_\delta$ is constructed from the original policy
$\eta_0$ by \emph{transferring} some transmissions from energy state $\epsilon$ to states $(\epsilon+1)$
and $(\epsilon-1)$, whereas transmissions in all other states are left unchanged.
Therefore, the new policy $\eta_\delta$
violates less P3.A, by diminishing the gap $\eta(\epsilon)-\eta(\epsilon+1)$ by a quantity $\delta+r(\delta)>0$.
 The functions  $r(\delta)>0$ and $h(\delta)\geq 0$ are uniquely defined as
\begin{enumerate}
 \item if $\epsilon=1$, let $h(\delta)=0$ and let $r(\delta)$ be such that
$\pi_{\eta_\delta}(e_{\mr{max}})=\pi_{\eta_0}(e_{\mr{max}}),\forall \delta<\kappa$,
 \item if $\epsilon>1$, let $h(\delta)$ and $r(\delta)$ be such that
\begin{align}\label{sdfsdfsf}
\left\{ \begin{array}{l}
  \pi_{\eta_\delta}(e_{\mr{max}})=\pi_{\eta_0}(e_{\mr{max}})\\
  \pi_{\eta_\delta}(0)=\pi_{\eta_0}(0)
 \end{array}\right.,\forall \delta<\kappa,
\end{align}
\end{enumerate}
where $0<\kappa\ll 1$ is an arbitrarily small constant, which guarantees an admissible 
policy $\eta_\delta\in\mathcal U$,
\emph{i.e.}, 
 $\eta_\delta(\epsilon-1)\in (0,1)$,
 $\eta_\delta(\epsilon)\in (0,1)$,
 $\eta_\delta(\epsilon+1)\in (0,1)$.

From Prop. \ref{SSD}, we have that $\xi_{\eta_\delta}(e)=\xi_{\eta_0}(e),\ \forall e>\epsilon+1$,
and therefore, from (\ref{pie}),
\begin{align}
&\pi_{\eta_\delta}(e)=
\frac{1}{\prod_{f=e}^{e_{\max}-1}\xi_{\eta_\delta}(f)}
\pi_{\eta_\delta}(e_{\max})
\nonumber\\&
=
\frac{1}{\prod_{f=e}^{e_{\max}-1}\xi_{\eta_0}(f)}
\pi_{\eta_0}(e_{\max})=\pi_{\eta_0}(e)
,\ e>\epsilon+1.
\end{align}
Similarly, $\xi_{\eta_\delta}(e)=\xi_{\eta_0}(e),\ \forall e<\epsilon-1$,
and therefore, from (\ref{pie}) and (\ref{sdfsdfsf}),
\begin{align}
&\pi_{\eta_\delta}(e)=
\left[\prod_{f=1}^{e-1}\xi_{\eta_\delta}(f)\right]\frac{\beta\pi_{\eta_\delta}(0)}{(1-\beta){\eta_\delta}(e)}
\nonumber\\&
=
\left[\prod_{f=1}^{e-1}\xi_{\eta_0}(f)\right]\frac{\beta\pi_{\eta_0}(0)}{(1-\beta){\eta_0}(e)}
=\pi_{\eta_0}(e)
,\ e<\epsilon-1.
\end{align}
It follows that $\pi_{\eta_\delta}(e)=\pi_{\eta_0}(e)$ for $e\in\mathcal E\setminus\{\epsilon-1,\epsilon,\epsilon+1\}$, \emph{i.e.},
the transmission transfer preserves the steady state distribution of visiting
the low and high energy states. 
 The derivative of $Z_\lambda(\eta_\delta)$ with respect to $\delta$, computed in $\delta{\to}0^+$,
is then given by
\begin{align}
\label{derivG}
&\left[\frac{\mathrm dZ_\lambda(\eta_\delta)}{\mathrm d\delta}\right]_{\delta\to 0^+}
\!\!\!\!\!\!\!\!=
\left[\frac{\mathrm d\pi_{\eta_{\delta}}(\epsilon-1)}{\mathrm d\delta}\right]_{\delta\to 0^+}\!\!\!\!\!\!\!\!z_\lambda(\eta_0(\epsilon-1))
\nonumber\\&
+\left[\frac{\mathrm d\pi_{\eta_{\delta}}(\epsilon)}{\mathrm d\delta}\right]_{\delta\to 0^+}\!\!\!\!\!\!\!\!z_\lambda(\eta_{0}(\epsilon))
+\left[\frac{\mathrm d\pi_{\eta_{\delta}}(\epsilon+1)}{\mathrm d\delta}\right]_{\delta\to 0^+}\!\!\!\!\!\!\!\!z_\lambda(\eta_0(\epsilon+1))
\nn\\&
+\pi_{\eta_{0}}(\epsilon-1)z_\lambda^\prime(\eta_0(\epsilon-1))h^\prime(0)
-\pi_{\eta_{0}}(\epsilon)z_\lambda^\prime(\eta_{0}(\epsilon))
\nn\\&
+\pi_{\eta_{0}}(\epsilon+1)z_\lambda^\prime(\eta_0(\epsilon+1))r^\prime(0),
\end{align}
where $h^\prime(0)=\left[\frac{\mathrm dh(\delta)}{\mathrm d\delta}\right]_{\delta\to 0^+}$
and $r^\prime(0)=\left[\frac{\mathrm dr(\delta)}{\mathrm d\delta}\right]_{\delta\to 0^+}$.

We first consider the case $\epsilon{<}e_{\max}{-}1$. The case $\epsilon{=}e_{\max}{-}1$ requires special treatment and will be discussed in the second part of the proof.
For the case $\epsilon<e_{\max}-1$,
we compute the derivative terms above.
Note that, using (\ref{pie})
and the fact that  $\pi_{\eta_{\delta}}(\epsilon-2)=\pi_{\eta_{0}}(\epsilon-2)$ and $\pi_{\eta_{\delta}}(\epsilon+2)=\pi_{\eta_{0}}(\epsilon+2)$
(if $\epsilon=1$, we define $\pi_{\eta_{\delta}}(-1)=\pi_{\eta_{0}}(-1)=0$),
 we have
\begin{align}\label{x1}
&\frac{\pi_{\eta_{0}}(\epsilon-2)}{\pi_{\eta_{0}}(\epsilon+2)}=
\frac{\pi_{\eta_{\delta}}(\epsilon-2)}{\pi_{\eta_{\delta}}(\epsilon+2)}
\\&\nn
{=}
\frac{(1{-}\beta)^4}{\beta^4}
\frac{\eta_\delta(\epsilon-1)}{1{-}\eta_\delta(\epsilon{-}2)}
\frac{\eta_\delta(\epsilon)}{1{-}\eta_\delta(\epsilon{-}1)}
\frac{\eta_\delta(\epsilon+1)}{1{-}\eta_\delta(\epsilon)}
\frac{\eta_\delta(\epsilon+2)}{1{-}\eta_\delta(\epsilon+1)}
.
\end{align}
By imposing the normalization $\sum_e\pi_{\eta_\delta}(e)=\sum_e\pi_{\eta_0}(e)=1$ and using (\ref{pie}), we have that
\begin{align}\label{x2}
&\frac{\pi_{\eta_{0}}(\epsilon{-}1){+}\pi_{\eta_{0}}(\epsilon){+}\pi_{\eta_{0}}(\epsilon{+}1)}{\pi_{\eta_{0}}(\epsilon+2)}=
\frac{\pi_{\eta_{\delta}}(\epsilon{-}1){+}\pi_{\eta_{\delta}}(\epsilon){+}\pi_{\eta_{\delta}}(\epsilon{+}1)}{\pi_{\eta_{\delta}}(\epsilon+2)}
\nn
\\&
=
\frac{(1-\beta)\eta_\delta(\epsilon+2)}{\beta(1-\eta_\delta(\epsilon+1))}
+
\frac{(1-\beta)\eta_\delta(\epsilon+1)}{\beta(1-\eta_\delta(\epsilon))}
\frac{(1-\beta)\eta_\delta(\epsilon+2)}{\beta(1-\eta_\delta(\epsilon+1))}
\nn
\\&
{+}
\frac{(1{-}\beta)\eta_\delta(\epsilon)}{\beta(1{-}\eta_\delta(\epsilon-1))}
\frac{(1-\beta)\eta_\delta(\epsilon+1)}{\beta(1-\eta_\delta(\epsilon))}
\frac{(1-\beta)\eta_\delta(\epsilon+2)}{\beta(1-\eta_\delta(\epsilon+1))}.
\end{align}
By computing the derivatives of (\ref{x1}) and (\ref{x2}) with respect to $\delta\to 0^+$,
 setting them equal to $0$ (since they are constant with respect to $\delta$),
and solving with respect to $r^\prime(0)$ with $h^\prime(0)$, we obtain
\begin{align}
&r^\prime(0)
=
\frac{(1-\beta)\eta_0(\epsilon+1)[1-\eta_0(\epsilon+1)]}{[1-\eta_0(\epsilon)][\beta+\eta_0(\epsilon)(1-2\beta)]},
\\&
h^\prime(0)=\frac{\beta\eta_0(\epsilon-1)[1-\eta_0(\epsilon-1)]}{\eta_0(\epsilon)[\beta+\eta_0(\epsilon)(1-2\beta)]}.
\label{hprime0}
 \end{align}
 Note that, if $\epsilon=1$, then (\ref{hprime0}) yields a feasible $h^\prime(0)=0$, coherent with (\ref{etamod}),
therefore such solution holds for $\epsilon=1$ as well.
Using the fact that $\pi_{\eta_{\delta}}(\epsilon-2)=\pi_{\eta_{0}}(\epsilon-2)$,
from (\ref{pie}) we obtain
\begin{align}
&\pi_{\eta_{\delta}}(\epsilon-2)=\frac{(1-\beta)\eta_\delta(\epsilon-1)}{\beta(1-\eta_\delta(\epsilon-2))}\pi_{\eta_{\delta}}(\epsilon-1)
\nn\\&
=\frac{(1-\beta)\eta_0(\epsilon-1)}{\beta(1-\eta_0(\epsilon-2))}\pi_{\eta_{0}}(\epsilon-1)
=\pi_{\eta_{0}}(\epsilon-2),\nn
\end{align}
and therefore, by computing the derivative of the second term above with respect to $\delta\to 0^+$,
and setting it equal to zero (since it is constant with respect to $\delta$),
 we obtain
\begin{align}
&\left[\frac{\mathrm d\pi_{\eta_{\delta}}(\epsilon-1)}{\mathrm d\delta}\right]_{\delta\to 0^+}
=
-h^\prime(0)\frac{1}{\eta_0(\epsilon-1)}\pi_{\eta_{0}}(\epsilon-1)
\nn\\&
=
-\frac{1-\beta}{\beta+\eta_0(\epsilon)(1-2\beta)}\pi_{\eta_{0}}(\epsilon).
\end{align}
Similarly,
from (\ref{pie}) and the fact that $\pi_{\eta_{\delta}}(\epsilon+2)=\pi_{\eta_{0}}(\epsilon+2)$, we have that
\begin{align}
&\pi_{\eta_{\delta}}(\epsilon+2)
=
\pi_{\eta_{\delta}}(\epsilon+1)
\frac{\beta(1-\eta_\delta(\epsilon+1))}{(1-\beta)\eta_\delta(\epsilon+2)}
\nn\\&
=
\pi_{\eta_{0}}(\epsilon+1)
\frac{\beta(1-\eta_0(\epsilon+1))}{(1-\beta)\eta_0(\epsilon+2)}
=
\pi_{\eta_{0}}(\epsilon+2)
,\nn
\end{align}
and therefore, by computing the derivative of the second term above with respect to $\delta\to 0^+$, 
and setting it equal to zero,
we obtain
\begin{align}
&\left[\frac{\mathrm d\pi_{\eta_{\delta}}(\epsilon+1)}{\mathrm d\delta}\right]_{\delta\to 0^+}
=r^\prime(0)\frac{1}{1-\eta_\delta(\epsilon+1)}\pi_{\eta_{0}}(\epsilon+1)
\nn\\&
=\frac{\beta}{\beta+\eta_0(\epsilon)(1-2\beta)}\pi_{\eta_{0}}(\epsilon).
\end{align}
Finally, using the fact that 
$\pi_{\eta_{\delta}}(\epsilon-1)+\pi_{\eta_{\delta}}(\epsilon)+\pi_{\eta_{\delta}}(\epsilon+1)
=
\pi_{\eta_{0}}(\epsilon-1)+\pi_{\eta_{0}}(\epsilon)+\pi_{\eta_{0}}(\epsilon+1)
$ (normalization),
we obtain
\begin{align}
&\left[\frac{\mathrm d\pi_{\eta_{\delta}}(\epsilon)}{\mathrm d\delta}\right]_{\delta\to 0^+}\!\!\!\!
=
-\left[\frac{\mathrm d\pi_{\eta_{\delta}}(\epsilon+1)}{\mathrm d\delta}\right]_{\delta\to 0^+}\!\!\!\!
-\left[\frac{\mathrm d\pi_{\eta_{\delta}}(\epsilon-1)}{\mathrm d\delta}\right]_{\delta\to 0^+}\!\!\!\!
\nn\\&
=
\frac{1-2\beta}{\beta+\eta_0(\epsilon)(1-2\beta)}\pi_{\eta_{0}}(\epsilon).
\end{align}
By substituting these expressions in (\ref{derivG}), we obtain
\begin{align}
&\frac{\beta{+}\eta_0(\epsilon)(1{-}2\beta)}{\pi_{\eta_{0}}(\epsilon)}
\left[\frac{\mathrm dZ_\lambda(\eta_\delta)}{\mathrm d\delta}\right]_{\delta\to 0^+}\!\!\!\!\!\!\!\!
=(1{-}2\beta)[g(\eta_{0}(\epsilon)){-}\lambda\eta_{0}(\epsilon)]
\nn\\&
-(1-\beta)[g(\eta_0(\epsilon-1))-\eta_0(\epsilon-1)g^\prime(\eta_0(\epsilon-1))]
\nn\\&
+\beta\{g(\eta_0(\epsilon+1))+g^\prime(\eta_0(\epsilon+1))[1-\eta_\delta(\epsilon+1)]-\lambda\}
\nn\\&
-[\beta+\eta_0(\epsilon)(1-2\beta)][g^\prime(\eta_{0}(\epsilon))-\lambda]
\nn\\&
\triangleq S(\eta_0(\epsilon-1),\eta_0(\epsilon),\eta_0(\epsilon+1)).
\end{align}
From the concavity of $g(x)$, $S(\eta_0(\epsilon-1),\eta_0(\epsilon),\eta_0(\epsilon+1))
$ is a decreasing function of $\eta_0(\epsilon+1)$ and $\eta_0(\epsilon-1)$. Then, since
$\eta_0(\epsilon+1)\leq \eta_0(\epsilon)$ and $\eta_0(\epsilon+1)<\eta_0(\epsilon)$
 by hypothesis, we have that 
\begin{align*}
&S(\eta_0(\epsilon-1),\eta_0(\epsilon),\eta_0(\epsilon+1))
\geq
S(\eta_0(\epsilon),\eta_0(\epsilon),\eta_0(\epsilon))=0,
\end{align*}
hence $\left[\frac{\mathrm d[G(\eta_\delta)-\lambda P(\eta_\delta)]}{\mathrm d\delta}\right]_{\delta\to 0^+}>0$.
Therefore, there exists a small $\kappa>\delta>0$ such that $\eta_\delta\in\mathcal U$
and
$Z_\lambda(\eta_\delta)>Z_\lambda(\eta_0)$,
hence $\eta_0$ is strictly sub-optimal.

Using the same approach for the case $\epsilon=e_{\max}-1$, 
and imposing the constraints  $\pi_{\eta_{\delta}}(\epsilon-2)=\pi_{\eta_{0}}(\epsilon-2)$
and $\pi_{\eta_{\delta}}(e_{\max})=\pi_{\eta_{0}}(e_{\max})$,
it can be proved that 
\begin{align*}
&
\frac{\beta+(1-\beta)\eta_0(\epsilon)}{\pi_{\eta_{0}}(\epsilon)}
\left[\frac{\mathrm dZ_\lambda(\eta_\delta)}{\mathrm d\delta}\right]_{\delta\to 0^+}\!\!\!\!\!\!\!
=\frac{\beta}{1-\beta}[g^\prime(\eta_0(\epsilon+1))-\lambda]
\nn\\&
{+}(1{-}\beta)\left[z_\lambda(\eta_{0}(\epsilon)){-}g(\eta_0(\epsilon-1)){+}\eta_0(\epsilon{-}1)g^\prime(\eta_0(\epsilon{-}1))\right]
\\&
-[g^\prime(\eta_{0}(\epsilon))-\lambda][\beta+(1-\beta)\eta_0(\epsilon)]
\nn\\&
\triangleq
T(\eta_0(\epsilon-1),\eta_0(\epsilon),\eta_0(\epsilon+1)).\nn
\end{align*}
Then, using the concavity of $g(x)$, 
we have that $T(\eta_0(\epsilon{-}1),\eta_0(\epsilon),\eta_0(\epsilon{+}1))$ is 
a decreasing function of $\eta_0(\epsilon+1)$ and $\eta_0(\epsilon-1)$.
Then, since $\eta_0(\epsilon+1)\leq \eta_0(\epsilon)$ and $\eta_0(\epsilon-1)<\eta_0(\epsilon)$ by hypothesis, we obtain
\begin{align}
&T(\eta_0(\epsilon-1),\eta_0(\epsilon),\eta_0(\epsilon+1))
>
T(\eta_0(\epsilon),\eta_0(\epsilon),\eta_0(\epsilon))
\nn\\&
=\frac{\beta^2}{1-\beta}[g^\prime(\eta_0(\epsilon))-\lambda].
\end{align}
Note that, if $g^\prime(\eta_0(\epsilon))-\lambda\geq 0$, then 
$\left[\frac{\mathrm d[G(\eta_\delta)-\lambda P(\eta_\delta)]}{\mathrm d\delta}\right]_{\delta\to 0^+}\!\!\!\!\!\!\!{>}0$,
hence there exists a small $\kappa>\delta>0$ such that $\eta_\delta\in\mathcal U$
and $G(\eta_\delta)-\lambda P(\eta_\delta)>G(\eta_0)-\lambda P(\eta_0)$, so that 
 $\eta_0$ is strictly sub-optimal.
 On the other hand, it can be proved using
 the transmission transfer technique that a policy $\eta$ such that 
 $g^\prime(\eta(e))-\lambda<0$ for some $e$,
 \emph{i.e.}, $\eta(e)>x_\lambda^*$,  is strictly sub-optimal. 
 An intuitive explanation is that under such policy, the reward in energy level $e$ satisfies the inequality
  $z_\lambda(\eta(e))<z_\lambda(x_\lambda^*)$.
Therefore, this policy can be improved by a policy $\tilde\eta$ which employs a smaller transmission probability 
$\tilde\eta(e)=x_\lambda^*<\eta(e)$ in state $e$ (under a proper transmission transfer to other states). Such policy  yields a larger instantaneous reward 
$z_\lambda(\tilde\eta(e))>z_\lambda(\eta(e))$ in state $e$, and a more conservative energy consumption,
thus yielding improved long-term performance.
The proof of this part is not included due to space constraints.
 
We now prove P3.B, \emph{i.e.}, $\eta^{(\lambda)}(e)\in (\eta_L(\lambda),\eta_H(\lambda)),\forall e{>}0$.
Using Prop. \ref{SSD},
the derivative of the steady state distribution with respect to $\eta(1)$ is given by
\begin{align}
&
 \frac{\mr{d}\pi_\eta(1)}{\mr d\eta(1)}=\frac{1}{1-\eta(1)}\pi_\eta(1)\left[1-\frac{\pi_\eta(1)}{\beta}\right],
 \nn\\&
 \frac{\mr{d}\pi_\eta(e)}{\mr d\eta(1)}=
-\pi_\eta(e)\pi_\eta(1)\frac{1}{\beta}\frac{1}{1-\eta(1)},\ e>1,
%
%
\end{align}
Then, from (\ref{GP}) we obtain
\begin{align}
\label{qweqwe}
& \frac{\mr{d} Z_\lambda(\eta)}{\mr d \eta(1)}=
\frac{\mr{d}\pi_\eta(1)}{\mr d \eta(1)}z_\lambda(\eta(1))+
\pi_\eta(1)[g^\prime(\eta(1))-\lambda]
\nn\\&
+\sum_{e=2}^{e_{\mr{max}}} \frac{\mr{d}\pi_\eta(e)}{\mr d \eta(1)}z_\lambda(\eta(e))
=
\frac{\pi_\eta(1)}{1-\eta(1)}\left(1-\frac{\pi_\eta(1)}{\beta}\right)z_\lambda(\eta(1))
\nn\\&
+\pi_\eta(1)[g^\prime(\eta(1))-\lambda]
-\frac{\pi_\eta(1)}{\beta(1-\eta(1))}
\sum_{e=2}^{e_{\mr{max}}}\pi_\eta(e)z_\lambda(\eta(e)).
\nn
\end{align}
Using the fact that $ \sum_{e=2}^{e_{\mr{max}}}\pi_\eta(e)z_\lambda(\eta(e)\!){=}Z_\lambda(\eta){-}\pi_\eta(1)z_\lambda(\eta(1)\!)$, we obtain
\begin{align*}
& \frac{\mr{d} Z_\lambda(\eta)}{\mr d \eta(1)}=
\pi_\eta(1)\frac{1}{1-\eta(1)}z_\lambda(\eta(1))
+\pi_\eta(1)[g^\prime(\eta(1))-\lambda]
\nn\\&
{-}\frac{\pi_\eta(1)}{1{-}\eta(1)}\frac{Z_\lambda(\eta)}{\beta}
{\propto}
g(\eta(1)){+}(1{-}\eta(1))g^\prime(\eta(1)){-}\lambda{-}\frac{Z_\lambda(\eta)}{\beta}\nonumber\\&
{>}
g(\eta(1)){+}(1{-}\eta(1))g^\prime(\eta(1)){-}\lambda{-}\frac{z_\lambda(\min\{x_\lambda^*,\beta\})}{\beta}{\triangleq}L(\eta(1)).
\end{align*}
Using the concavity of $g(x)$, it can be shown that $L(\eta(1))$ is a decreasing function of
$\eta(1)$, with $\lim_{\eta(1)\to 0^+}L(\eta(1))=\infty$.
Moreover, using the fact that $g^\prime(x_\lambda^*)-\lambda=0 $ if $x_\lambda^*<\beta$,
and $z_\lambda(\min\{x_\lambda^*,\beta\})>z_\lambda(0)=0$ from the definition of $x_\lambda^*$ and the concavity of $g(x)$,
we obtain $L(\min\{x_\lambda^*,\beta\})=-z_\lambda(\min\{x_\lambda^*,\beta\})\frac{1-\beta}{\beta}<0$.
 Therefore, there exists a unique $\eta_L(\lambda)\in (0,\min\{x_\lambda^*,\beta\})$
that solves $L(\eta_L(\lambda))=0$.
Then, for all $\eta(1)\leq \eta_L(\lambda)$ we have $L(\eta(1))\geq 0$, hence
$ \frac{\mr{d}Z_\lambda(\eta)}{\mr d\eta(1)}>0$,
which proves that $\eta(1)\leq \eta_L(\lambda)$ is strictly suboptimal.

Similarly, the derivative of the steady state distribution with respect to $\eta(e_{\mr{max}})$ is given by
\begin{align}
&\frac{\mr{d}\pi_\eta(e)}{\mr d\eta(e_{\mr{max}})}=\frac{\pi_\eta(e)\pi_\eta(e_{\mr{max}})}{\eta(e_{\mr{max}})},\ e<e_{\mr{max}},
\\&
\frac{\mr{d}\pi_\eta(e_{\mr{max}})}{\mr d\eta(e_{\mr{max}})}=
-\frac{\pi_\eta(e_{\mr{max}})}{\eta(e_{\mr{max}})}(1-\pi_\eta(e_{\mr{max}})).
\end{align}
Then, from (\ref{GP}) we obtain
\begin{align*}
& \frac{\mr{d} Z_\lambda(\eta)}{\mr d \eta(e_{\mr{max}})}{=}
 \frac{\mr{d}\pi_\eta(e_{\mr{max}})}{\mr d \eta(e_{\mr{max}})}z_\lambda(\eta(e_{\mr{max}})\!)
{+}\!\pi_\eta(e_{\mr{max}})[g^\prime(\eta(e_{\mr{max}})\!){-}\lambda]
 \nn\\&
{+}\!\!\!\!\!\sum_{e=1}^{e_{\mr{max}}-1}\!\!\!\!\!\frac{\mr{d}\pi_\eta(e)}{\mr d \eta(e_{\mr{max}})}z_\lambda(\eta(e)\!)
{=}{-}\frac{\pi_\eta(e_{\mr{max}})}{\eta(e_{\mr{max}})}(1{-}\pi_\eta(e_{\mr{max}})\!)z_\lambda(\eta(e_{\mr{max}})\!)
\\&
{+}\pi_\eta(e_{\mr{max}})[g^\prime(\eta(e_{\mr{max}})){-}\lambda]
{+}\frac{\pi_\eta(e_{\mr{max}})}{\eta(e_{\mr{max}})}
\!\sum_{e=1}^{e_{\mr{max}}-1}\!\!\!\pi_\eta(e)z_\lambda(\eta(e)).\nn
\end{align*}
Note that $\sum_{e=1}^{e_{\mr{max}}-1}\!\!\!\!\pi_\eta(e)z_\lambda(\eta(e)){=}Z_\lambda(\eta){-}\pi_\eta(e_{\mr{max}})z_\lambda(\eta(e_{\mr{max}}))$.
Substituting we obtain
\begin{align}
& \frac{\mr{d}Z_\lambda(\eta)}{\mr d \eta(e_{\mr{max}})}{=}
{-}\frac{\pi_\eta(e_{\mr{max}})}{\eta(e_{\mr{max}})}z_\lambda(\eta(e_{\mr{max}})\!)
{+}\pi_\eta(e_{\mr{max}})[g^\prime(\eta(e_{\mr{max}})\!){-}\lambda]
\nn\\&
{+}\frac{\pi_\eta(e_{\mr{max}})}{\eta(e_{\mr{max}})}Z_\lambda(\eta)
{\propto}
{-}g(\eta(e_{\mr{max}})){+}\eta(e_{\mr{max}})g^\prime(\eta(e_{\mr{max}}))
{+}Z_\lambda(\eta)\nonumber\\&
<
-g(\eta(e_{\mr{max}}))+\eta(e_{\mr{max}})g^\prime(\eta(e_{\mr{max}}))
+z_\lambda(\min\{x_\lambda^*,\beta\})
\nn\\&
\triangleq U(\eta(e_{\mr{max}})).
\end{align}
Since $g(x)$ is concave, $U(\eta(e_{\mr{max}}))$ is a decreasing function of
$\eta(e_{\mr{max}})$.
If $x_\lambda^*\leq\beta$, using the fact that $g^\prime(x_\lambda^*)=\lambda$, we obtain $U(x_\lambda^*)=0$.
On the other hand, if $x_\lambda^*>\beta$, we have that 
$U(x_\lambda^*)=-[g(x_\lambda^*)+g^\prime(x_\lambda^*) (\beta-x_\lambda^*)-g(\beta)]<0$
and $U(\beta)=\beta[g^\prime(\beta)-\lambda]>0$.
Therefore, there exists a unique $\eta_U(\lambda)\in (\min\{\beta,x_\lambda^*\},x_\lambda^*)$ that solves
$U(\eta_U(\lambda))=0$.
 Then, for all $\eta(e_{\mr{max}})\geq \eta_U(\lambda)$ we have
$U(\eta(e_{\mr{max}}))\leq 0$, hence $\frac{\mr{d}Z_\lambda(\eta)}{\mr d\eta(e_{\mr{max}})}{<}0$,
which proves that $\eta(e_{\mr{max}})\geq \eta_U(\lambda)$ is strictly suboptimal.
Finally, by combining P3.A and P3.B, we obtain
$\eta_L(\lambda)<\eta(1)<\eta(2)<\dots<\eta(e_{\mr{max}})<\eta_U(\lambda)$.
The theorem is thus proved.
\hfill$\blacksquare$
\section*{Appendix~C: Proof of Algorithm \ref{PIA}}
For the policy evaluation step,
we need to prove that the \emph{relative value function} \cite{Bertsekas2005}, denoted as $v_\eta(e)$, takes the form (\ref{RVF}), where we have defined 
$D_\eta(e)=v_\eta(e)-v_\eta(e-1),\ e>0$, $D_\eta(0)=0$.
By the definition of relative value function \cite{Bertsekas2005}, $v_\eta(e)$ is obtained as the unique solution of the linear system
\begin{align*}
&v_\eta(0){=}0,\ v_\eta(e){-}\sum_{f\in\mathcal E}\mathbb P_\eta(f|e)v_\eta(f)
 {=}z_\lambda(\eta(e)){-}Z_\lambda(\eta),e{\in}\mathcal E,
\end{align*}
where $\mathbb P_\eta(f|e)$ is the transition probability from energy level $e$ to energy level $f$, under policy $\eta$.
In particular, by explicating $\mathbb P_\eta(f|e)$,
the linear system above for $e{<}e_{\max}$ is equivalent to
\begin{align}
D_\eta(e+1)=\frac{Z_\lambda(\eta)-z_\lambda(\eta(e))}{\beta(1-\eta(e))}+\frac{(1-\beta)\eta(e)}{\beta(1-\eta(e))}D_\eta(e),
\end{align}
proving the recursion (\ref{RVF}).
In the policy improvement step, we solve the optimization problem
\begin{align*}
\eta^{[i+1]}(e){=}\underset{\eta(e)\in (\eta_L,\eta_H)}{\arg\max}
z_\lambda(\eta(e)){+}\sum_f\mathbb P_\eta(f|e)v^{[i]}(f),\forall e{\neq}0.
 \end{align*}
 In particular, by explicating the transition probability $\mathbb P_\eta(f|e)$, we obtain
\begin{align}
\eta^{[i+1]}(e)= \arg\max_{\eta(e)\in (\eta_L,\eta_H)} g(\eta(e))-x^{[i]}(e,\lambda)\eta(e),
\end{align}
 where
  $x^{[i]}(e,\lambda)$ is given by (\ref{rho}). 
  The solution (\ref{maxPI}) directly follows by
  exploiting the strict concavity of $g(x)$ and the fact that $g^\prime(x)=\bar V(y_{\mathrm{th}}(x))$. 
\hfill$\blacksquare$
\section*{Appendix~D: Proof of Theorem \ref{thm1}}
In order to prove the theorem, we need the following propositions.
\begin{propos}\label{Pdecreasing}
 $P(\eta^{(\lambda)})$ is a non-increasing function of $\lambda$, and
$P(\eta^{(0)})>0$,
$\lim_{\lambda\to \infty}P(\eta^{(\lambda)})=0$.
\end{propos}
\begin{proof}
Assume by contradiction that $\lambda_1>\lambda_2$ and 
$P(\eta^{(\lambda_2)})<P(\eta^{(\lambda_1)})$.
Then we have
\begin{align}
&Z_{\lambda_2}(\eta^{(\lambda_2)})
\geq
 Z_{\lambda_1}(\eta^{(\lambda_1)})
+(\lambda_1-\lambda_2) P(\eta^{(\lambda_1)})
\label{k0}
\\&
\label{k1}
>
  Z_{\lambda_1}(\eta^{(\lambda_1)})
+(\lambda_1-\lambda_2) P(\eta^{(\lambda_2)})
\\
\label{k4}
&\geq
 G(\eta^{(\lambda_2)})-\lambda_1 P(\eta^{(\lambda_2)})
+(\lambda_1-\lambda_2) P(\eta^{(\lambda_2)})
\\&\nn
=G(\eta^{(\lambda_2)})-\lambda_2 P(\eta^{(\lambda_2)}),
\end{align}
where in (\ref{k0})-(\ref{k1})
we have used the fact that 
$Z_{\lambda_2}(\eta^{(\lambda_2)})
\geq 
Z_{\lambda_2}(\eta^{(\lambda_1)})
$ and then the hypothesis,
in (\ref{k4})
the fact that 
$
Z_{\lambda_1}(\eta^{(\lambda_1)})
\geq
Z_{\lambda_1}(\eta^{(\lambda_2)})
$.
By comparing (\ref{k1}) and (\ref{k4}), we obtain the contradiction
$Z_{\lambda_2}(\eta^{(\lambda_2)})>Z_{\lambda_2}(\eta^{(\lambda_2)})$.
For the second part of the proposition, since $\eta^{(0)}\neq\{\text{idle policy}\}$, then $P(\eta^{(0)})>0$.
Since
 $P(\eta^{(\lambda)})<\eta_U(\lambda)\leq x_\lambda^*$ (P3 of Theorem \ref{propeta}),
where $x_\lambda^*=\arg\max_{x}[g(x)-\lambda x]$, in the limit $\lambda\to\infty$ we obtain
$x_\lambda^*\to 0$, hence $P(\eta^{(\lambda)})\to 0$.
The proposition is thus proved.
\end{proof}
\begin{propos}\label{Lembdadecreasing}
 $\Lambda(\eta^{(\lambda)})$ is a continuous, non-increasing function of 
$\lambda$, for $\lambda\geq 0$, with limits
$\Lambda(\eta^{(0)})\in (0,\infty)$
and
$\lim_{\lambda\to \infty}\Lambda(\eta^{(\lambda)})=0$.
$h(\lambda)\triangleq\Lambda(\eta^{(\lambda)})-\lambda$ is a continuous, strictly decreasing function of $\lambda$, with limits
$h(0)>0$
and
$\lim_{\lambda\to \infty}h(\lambda)=-\infty$.
\end{propos}
\begin{proof}
The continuity of  $\Lambda(\eta^{(\lambda)})$ follows from the continuity of $\eta^{(\lambda)}$ (P2 of Theorem \ref{propeta})
and the definition of $\Lambda(\eta)$ in (\ref{Lambda}).
For $\lambda=0$ we have $\eta^{(0)}=\arg\max_{\eta\in\mathcal U}G(\eta)$.
Then, we obtain
\begin{align}
\Lambda(\eta^{(0)})=(U-1)\frac{G(\eta^{(0)})}{1-P(\eta^{(0)})}\in (0,\infty).
\end{align}
$\Lambda(\eta^{(0)})$ is positive and bounded since $0{<}G(\eta^{(0)}){\leq}g(1){<}\infty$
and $P(\eta^{(0)})\leq\beta<1$.
On the other hand, for $\lambda\to\infty$, we have $\eta^{(\lambda)}\to 0$,
hence $G(\eta^{(\lambda)})\to 0 $, $P(\eta^{(\lambda)})\to 0$
and $\lim_{\lambda\to \infty}\Lambda(\eta^{(\lambda)})=0$.
To conclude, we prove that 
 $\Lambda(\eta^{(\lambda)})$ is a non-increasing function of 
$\lambda$, \emph{i.e.},
$\Lambda(\eta^{(\lambda_1)})\leq\Lambda(\eta^{(\lambda_2)})$
for $\lambda_1>\lambda_2\geq 0$.
Using (\ref{Lambda}), this is true if and only if
\begin{align*}
G(\eta^{(\lambda_2)})(1-P(\eta^{(\lambda_1)}))
-G(\eta^{(\lambda_1)})(1-P(\eta^{(\lambda_2)}))
\geq 0.
\end{align*}
Equivalently, by rearranging the terms,
\begin{align*}
&A(\lambda_1,\lambda_2){\triangleq}
Z_{\lambda_2}(\eta^{(\lambda_2)})[1{-}P(\eta^{(\lambda_1)})]
{-}Z_{\lambda_1}(\eta^{(\lambda_1)})[1{-}P(\eta^{(\lambda_2)})]
\\&
+\lambda_2P(\eta^{(\lambda_2)})[1-P(\eta^{(\lambda_1)})]
-\lambda_1P(\eta^{(\lambda_1)})[1-P(\eta^{(\lambda_2)})]
\geq 0.
\end{align*}
Using the fact that $\eta^{(\lambda_2)}$ is optimal for (\ref{etalambda}) for $\lambda=\lambda_2$,
hence $Z_{\lambda_2}(\eta^{(\lambda_2)})
\geq Z_{\lambda_2}(\eta^{(\lambda_1)})$,
a sufficient condition which guarantees that $A(\lambda_1,\lambda_2)\geq 0$ is that
\begin{align*}
&A(\lambda_1,\lambda_2){\geq}Z_{\lambda_2}(\eta^{(\lambda_1)})[1{-}P(\eta^{(\lambda_1)})]
{-}Z_{\lambda_1}(\eta^{(\lambda_1)})
[1{-}P(\eta^{(\lambda_2)})]
\\
&{+}\lambda_2P(\eta^{(\lambda_2)})(1{-}P(\eta^{(\lambda_1)}))
{-}\lambda_1P(\eta^{(\lambda_1)})(1{-}P(\eta^{(\lambda_2)}))
{\geq}0.
\end{align*}
By rearranging the terms,
it can be readily verified that the second term above is equivalent to
\begin{align*}
[P(\eta^{(\lambda_2)})-P(\eta^{(\lambda_1)})][G(\eta^{(\lambda_1)})+\lambda_2(1-P(\eta^{(\lambda_1)}))]
\geq 0,
\end{align*}
which holds from 
 Prop. \ref{Pdecreasing}.
The second part of the proposition readily follows.
\end{proof}

Using Prop. \ref{Lembdadecreasing}, we have that 
 there exists a unique $\lambda^*\in (0,\infty)$ such that 
$h(\lambda^*)=0$, \emph{i.e.}, $\Lambda(\eta^{(\lambda^*)})=\lambda^*$.
Under such $\lambda^*$, $\eta^{(\lambda^*)}$ is optimal for (\ref{NE}),
and $\eta^{(\lambda^*)}$ is unique from P1 of Theorem~\ref{propeta}.

We now prove that under the SNE, $P(\eta^{(\lambda^*)}){\leq}\frac{1}{U}$
($P(\eta^{(\lambda^*)})\leq\beta$ as shown in the proof of Prop. \ref{upbound}).
This is true if $P(\eta^{(0)})\leq\frac{1}{U}$,
since $P(\eta^{(\lambda)})$ is a non-increasing function of $\lambda$ (Prop. \ref{Pdecreasing}).
Now, assume that  $P(\eta^{(0)})>\frac{1}{U}$.
Since 
 $\underset{\lambda\to\infty}\lim P(\eta^{(\lambda)})=0$ (Prop. \ref{Pdecreasing}),
there exists $\hat\lambda\in (0,\infty)$ such that 
$P(\eta^{(\hat\lambda)})=\frac{1}{U}$.
For such $\hat\lambda$, from (\ref{etalambda}) we have
\begin{align}\label{z2}
G(\eta^{(\hat\lambda)})-\hat\lambda\frac{1}{U}{=}
Z_{\hat\lambda}(\eta^{(\hat\lambda)})
{=}
\max_{\eta\in\mathcal U}Z_{\hat\lambda}(\eta)
{\geq}
Z_{\hat\lambda}(\text{idle}){=}0,
\end{align}
where ``$\text{idle}$" denotes the policy where the EHS remains always idle, resulting in $G(\text{idle})=P(\text{idle})=0$.
The inequality follows from the fact that $Z_\lambda(\eta)$ is evaluated under the specific policy ``idle."
Then, using (\ref{z2}), (\ref{Lambda}) and the fact that $P(\eta^{(\hat\lambda)})=\frac{1}{U}$ by definition of $\hat\lambda$, we obtain
\begin{align*}
UG(\eta^{(\hat\lambda)})-\hat\lambda
=
(U-1)\frac{G(\eta^{(\hat\lambda)})}{1-1/U}
-\hat\lambda
=
\Lambda(\eta^{(\hat\lambda)})
-\hat\lambda
\geq 
0.
\end{align*}
Therefore, we obtain $\Lambda(\eta^{(\hat\lambda)})\geq\hat\lambda$.
Since $\Lambda(\eta^{(\lambda)})-\lambda$
is a decreasing function of $\lambda$ (Prop. \ref{Lembdadecreasing})
and the Lagrange multiplier $\lambda^*$ of the SNE must satisfy  $\Lambda(\eta^{(\lambda^*)})-\lambda^*=0$,
necessarily $\hat\lambda\leq\lambda^*$.
Finally, using Prop. \ref{Pdecreasing},
we obtain $P(\eta^{(\hat\lambda)})=\frac{1}{U}\geq P(\eta^{(\lambda^*)})$,
thus proving the theorem.
\hfill$\blacksquare$
\section*{Appendix~E: Proof of Theorem \ref{thm2}}
From Theorem~\ref{propeta}, we have that $\eta^*(e)\in(\eta_L(\lambda^*),\eta_H(\lambda^*))$ where
$0<\eta_L(\lambda^*)<\eta_H(\lambda^*)<1$. 
It follows that $\eta^*\in\mathrm{int}(\mathcal U)$.
Therefore,
since $\eta^*$ is globally optimal for the optimization problem (\ref{NE}) and it is unique,
 the gradient with respect to $\eta$ of the objective function in (\ref{NE}), computed in $\eta^*$ and denoted as
 $\Delta_{\eta^*}(\cdot)$,
 is equal to zero, and its Hessian
with respect to $\eta$, computed in $\eta^*$ and denoted as $\mathbf H_{\eta^*}(\cdot)$, is 
\emph{negative} definite, \emph{i.e.},
for the SNE $\eta^*$ we have
\begin{align}
\label{deltalo}
&\Delta_{\eta^*}(G(\eta))-\Lambda(\eta^*)\Delta_{\eta^*}(P(\eta))=\mathbf 0,
\\
\label{Hlo}
&\mathbf H_{\eta^*}(G(\eta))-\Lambda(\eta^*)\mathbf H_{\eta^*}(P(\eta))\prec 0.
\end{align}
On the other hand, the gradient of the network utility (\ref{rewtot}) is given by
\begin{align}\label{deltaglobal}
& \Delta_\eta\left(R_{\bs\eta}\right)
=
U(1-P(\eta))^{U-1} \Delta_\eta\left(G(\eta)\right)
\nn\\&
-U(U-1)G(\eta)(1-P(\eta))^{U-2}\Delta_\eta\left(P(\eta)\right).
\end{align}
The Hessian matrix of (\ref{rewtot})  is then obtained by further computing the gradient of (\ref{deltaglobal}),
yielding
\begin{align}
\label{H}
&\mathbf H_\eta\left(R_{\bs\eta}\right)
=
U(1-P(\eta))^{U-1}\mathbf H_\eta\left(G(\eta)\right)
\\&
-U(U-1)(1-P(\eta))^{U-2}\Delta_\eta\left(G(\eta)\right)\Delta_\eta(P(\eta))^T
\nonumber\\&
-U(U-1)(1-P(\eta))^{U-2}\Delta_\eta\left(P(\eta)\right)\Delta_\eta(G(\eta))^T
\nn\\&
-U(U-1)G(\eta)(1-P(\eta))^{U-2}\mathbf H_\eta\left(P(\eta)\right)
\nonumber\\&
{+}U(U{-}1)(U{-}2)G(\eta)(1{-}P(\eta))^{U-3}\Delta_\eta\left(P(\eta)\right)\Delta_\eta(P(\eta))^T.
\nonumber
\end{align}
By computing (\ref{deltaglobal}) under the SNE $\eta^*$, using (\ref{Lambda}) and substituting (\ref{deltalo}) in (\ref{deltaglobal}),
 we obtain $\Delta_{\eta^*}\left(R_{\bs\eta}\right)=\mathbf 0$.
Moreover, 
since $\mathbf H_{\eta^*}(G(\eta))\prec \Lambda(\eta^*)\mathbf H_{\eta^*}(P(\eta))$ from (\ref{Hlo}),
from (\ref{H})
we obtain
\begin{align*}
&\mathbf H_{\eta^*}\left(R_{\bs\eta}\right)
\prec
-U^2(U-1)G(\eta^*)
(1-P(\eta))^{U-3}
\nn\\&
\times
\Delta_{\eta^*}(P(\eta))\Delta_{\eta^*}(P(\eta))^T
\preceq 0.
\end{align*}
Therefore, $\mathbf H_{\eta^*}\left(R_{\bs\eta}\right)\prec 0$
and $\Delta_{\eta^*}(R_{\bs\eta})=\mathbf 0$,
hence $\eta^*$ is a local optimum for (\ref{origopt}).
\hfill$\blacksquare$

\bibliographystyle{IEEEtran}
\bibliography{IEEEabrv,References} 

\begin{IEEEbiography}[{\includegraphics[width=1in,height=1.25in,clip,keepaspectratio]{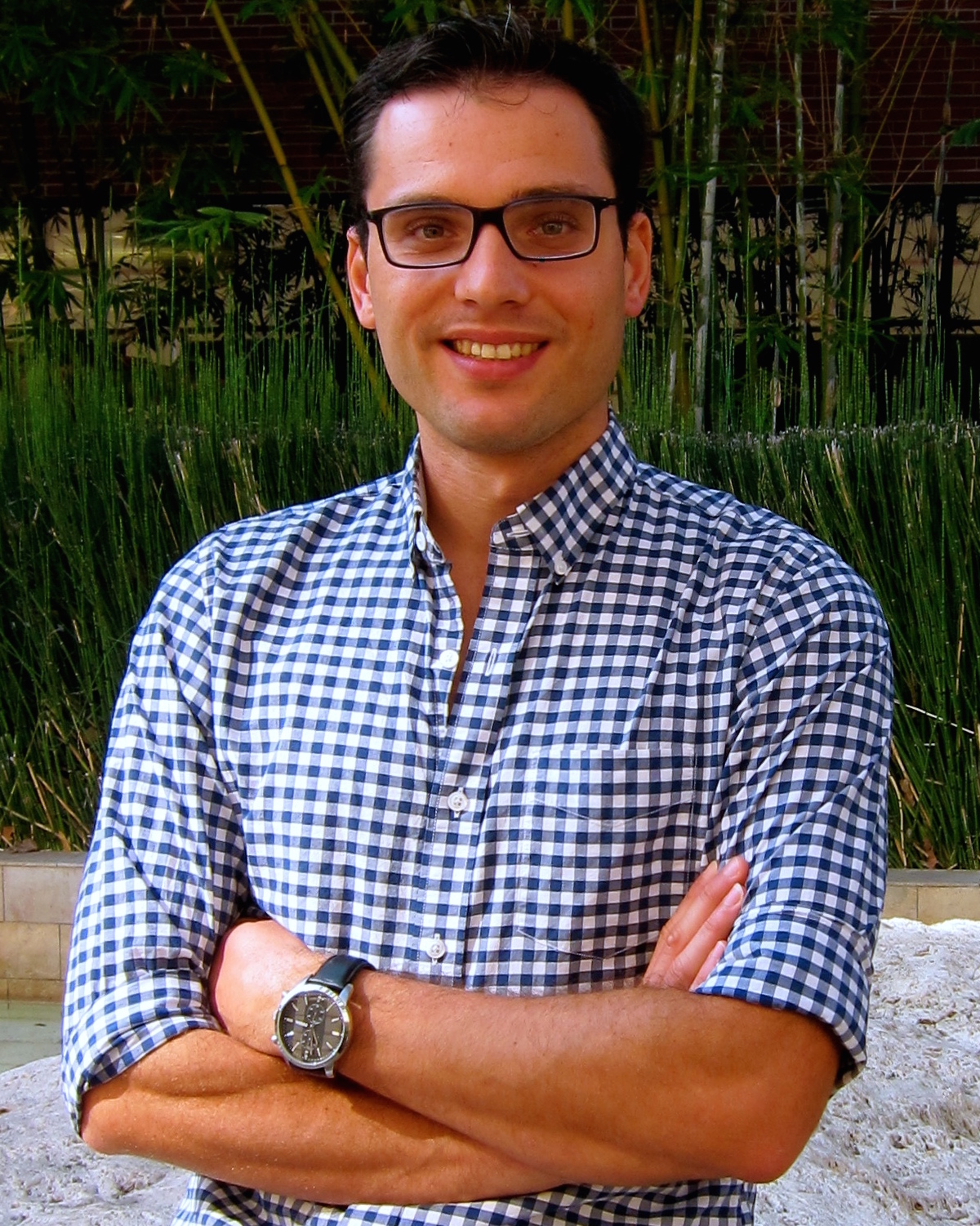}}]
{Nicol\`{o}~Michelusi} (S'09, M'13) received the B.Sc. (with honors), M.Sc.
 (with honors) and Ph.D. degrees from the University of Padova, Italy, in 2006, 2009 and 2013, respectively,
and the M.Sc. degree in Telecommunications Engineering from the Technical University of Denmark in 2009, as part of the T.I.M.E. double degree program.
In 2011, he was at the University
of Southern California, Los Angeles, USA,
and, in Fall 2012, at Aalborg University, Denmark, as a visiting research scholar.
He is currently a post-doctoral research fellow at the Ming Hsieh Department of Electrical Engineering, University of Southern California, USA.
 His research interests lie in the areas of
wireless networks, stochastic optimization, distributed estimation and modeling of bacterial networks.
  Dr. Michelusi  serves as a reviewer for the IEEE Transactions on Communications, IEEE Transactions on Wireless Communications, IEEE Transactions on Information Theory, IEEE Transactions on Signal Processing, IEEE Journal on Selected Areas in Communications,
 and IEEE/ACM Transactions on Networking.
 \end{IEEEbiography}

\begin{IEEEbiography}[{\includegraphics[width=1in,height=1.25in,clip,keepaspectratio]{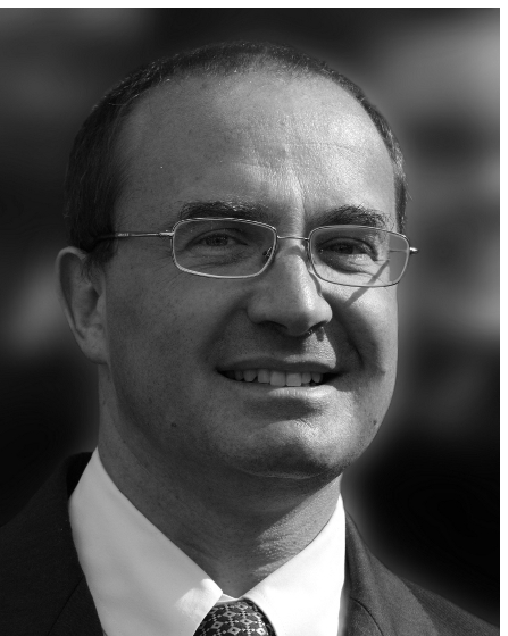}}]
{Michele Zorzi} (S'89, M'95, SM'98, F'07) received his Laurea and PhD degrees in electrical engineering from the University of Padova in 1990 and 1994, respectively. During academic year 1992-1993 he was on leave at UCSD, working on multiple access in mobile radio networks. In 1993 he joined the faculty of the Dipartimento di Elettronica e Informazione, Politecnico di Milano, Italy. After spending three years with the Center for Wireless Communications at UCSD, in 1998 he joined the School of Engineering of the University of Ferrara, Italy, where he became a professor in 2000. Since November 2003 he has been on the faculty of the Information Engineering Department at the University of Padova. His present research interests include performance evaluation in mobile communications systems, random access in wireless networks, ad hoc and sensor networks, Internet-of-Things, energy constrained communications and protocols, cognitive networks, and underwater communications and networking.

He was Editor-In-Chief of IEEE Wireless Communications from 2003 to 2005 and Editor-In-Chief of the IEEE Transactions on Communications from 2008 to 2011, 
and is currently the Editor-in-Chief of the new IEEE Transactions on Cognitive Communications and Networking. He has been an Editor for several journals and a member of the Organizing or the Technical Program Committee for many international conferences. He was also guest editor for special issues in IEEE Personal Communications (``Energy Management in Personal Communications Systems''), IEEE Network (``Video over Mobile Networks'') and IEEE Journal on Selected Areas in Communications (``Multimedia Network Radios,'' ``Underwater Wireless Communications and Networking'' and ``Energy Harvesting and Energy Transfer''). He served as a Member-at-Large of the Board of Governors of the IEEE Communications Society from 2009 to 2011, and is currently its Director of Education.
\end{IEEEbiography}

\end{document}